\documentclass[sigconf]{acmart}

\usepackage{algorithm,algpseudocode}
\usepackage{graphicx}
\usepackage{textcomp}
\usepackage{xspace}
\usepackage{url}
\usepackage[utf8]{inputenc}
\usepackage[english]{babel}
\usepackage{lipsum}
\usepackage{dblfloatfix}
\usepackage{enumitem, array}
\usepackage{subfigure}
\usepackage{tablefootnote}
\usepackage{xcolor}
\usepackage{colortbl}
\usepackage{setspace}
\usepackage[utf8]{inputenc}
\usepackage[english]{babel}
\usepackage{multicol}
\usepackage{multirow}
\usepackage{tikz}
\usepackage{booktabs}

\usepackage[colorinlistoftodos]{todonotes}
\usepackage{soul}
\newcommand{\remove}[1]{}

\usepackage{pifont}
\newcommand{\cmark}{\ding{51}}%
\newcommand{\xmark}{\ding{55}}%

\newcommand{\DBGePh}{\mathsf{GDB}}

\newcommand{\DOID}{\mathsf{ID}}
\newcommand{\DO}{O}
\newcommand{\Trustee}{\mathcal{T}}
\newcommand{\Vetter}{\mathcal{V}}
\newcommand{\Server}{\mathcal{D}}

\newcommand{\ACEPrivGenDB}{\mathsf{ACE}}

\newcommand{\User}{\mathcal{U}}

\newcommand{\K}{\mathsf{K}}
\definecolor{mygray1}{gray}{0.6}
\definecolor{mygray2}{gray}{0.8}

\definecolor{orange}{cmyk}{0.1,0.5,0.9,0.3}

\usepackage{graphicx}
\usepackage{textcomp}
\usepackage{xcolor}

\AtBeginDocument{%
  \providecommand\BibTeX{{%
    \normalfont B\kern-0.5em{\scshape i\kern-0.25em b}\kern-0.8em\TeX}}}
\acmYear{2023}
\acmDOI{XXXXXXX.XXXXXXX}

\remove{
\setcopyright{acmcopyright}
\copyrightyear{2018}
\acmYear{2023}
\acmDOI{XXXXXXX.XXXXXXX}

\acmConference[Conference acronym 'XX]{Make sure to enter the correct
  conference title from your rights confirmation emai}{June 03--05,
  2018}{Woodstock, NY}
\acmPrice{15.00}
\acmISBN{978-1-4503-XXXX-X/18/06}

}

\begin{document}

\title{ACE: A Consent-Embedded privacy-preserving search on genomic database}

\remove{

Sara Jafarbeiki (Monash University and CSIRO's Data61);
  Amin Sakzad (Monash University);
  Ron Steinfeld (Monash University);
  Shabnam Kasra Kermanshahi (RMIT);
  Chandra Thapa (CSIRO's Data61);
  Yuki Kume (Monash University)

}

\author{Sara Jafarbeiki$^{1,2,*}$, Amin Sakzad$^{1}$, Ron Steinfeld$^{1}$, Shabnam Kasra Kermanshahi$^{3}$, Chandra Thapa$^{2}$, Yuki Kume$^{1}$} 
    \affiliation{ 
      \institution{}
      \streetaddress{}
      \city{$^{1}$ Monash University, $^{2}$ CSIRO's Data61, $^{3}$ University of New South Wales (UNSW) Canberra} 
      \state{} 
      \country{Australia}
      \postcode{}
    }
    \email{* sara.jafarbeiki@monash.edu}

\remove{
\author{Ben Trovato}
\authornote{Both authors contributed equally to this research.}
\email{trovato@corporation.com}
\orcid{1234-5678-9012}
\author{G.K.M. Tobin}
\authornotemark[1]
\email{webmaster@marysville-ohio.com}
\affiliation{%
  \institution{Institute for Clarity in Documentation}
  \streetaddress{P.O. Box 1212}
  \city{Dublin}
  \state{Ohio}
  \country{USA}
  \postcode{43017-6221}
}

\author{Lars Th{\o}rv{\"a}ld}
\affiliation{%
  \institution{The Th{\o}rv{\"a}ld Group}
  \streetaddress{1 Th{\o}rv{\"a}ld Circle}
  \city{Hekla}
  \country{Iceland}}
\email{larst@affiliation.org}

\author{Valerie B\'eranger}
\affiliation{%
  \institution{Inria Paris-Rocquencourt}
  \city{Rocquencourt}
  \country{France}
}

\author{Aparna Patel}
\affiliation{%
 \institution{Rajiv Gandhi University}
 \streetaddress{Rono-Hills}
 \city{Doimukh}
 \state{Arunachal Pradesh}
 \country{India}}

\author{Huifen Chan}
\affiliation{%
  \institution{Tsinghua University}
  \streetaddress{30 Shuangqing Rd}
  \city{Haidian Qu}
  \state{Beijing Shi}
  \country{China}}

\author{Charles Palmer}
\affiliation{%
  \institution{Palmer Research Laboratories}
  \streetaddress{8600 Datapoint Drive}
  \city{San Antonio}
  \state{Texas}
  \country{USA}
  \postcode{78229}}
\email{cpalmer@prl.com}

\author{John Smith}
\affiliation{%
  \institution{The Th{\o}rv{\"a}ld Group}
  \streetaddress{1 Th{\o}rv{\"a}ld Circle}
  \city{Hekla}
  \country{Iceland}}
\email{jsmith@affiliation.org}

\author{Julius P. Kumquat}
\affiliation{%
  \institution{The Kumquat Consortium}
  \city{New York}
  \country{USA}}
\email{jpkumquat@consortium.net}
}
\renewcommand{\shortauthors}{}

\begin{abstract}
In this paper, we introduce ACE, a consent-embedded searchable encryption scheme. ACE enables dynamic consent management by supporting the physical deletion of associated data at the time of consent revocation. This ensures instant real deletion of data, aligning with privacy regulations and preserving individuals' rights. 
We evaluate ACE in the context of genomic databases
, demonstrating its ability to perform the addition and deletion of genomic records and related information based on ID, which especially complies with the requirements of deleting information of a particular data owner. To formally prove that ACE is secure under non-adaptive attacks, we present two new definitions of forward and backward privacy. We also define a new hard problem, which we call D-ACE, that facilitates the proof of our theorem (we formally prove its hardness by a security reduction from DDH to D-ACE). We finally present implementation results to evaluate the performance of ACE.

\remove{
In this paper, we present a consent-embedded searchable encryption based scheme, which we call ACE. This scheme is presented and evaluated considering genomic database.
Moreover, the embedded consent is dynamic, which prevents the search over the associated genomic data whenever the consent is revoked. Therefore, a dynamic system model that supports the physical deletion of data at the time of consent revocation is required to provide a dynamic consent feature.
In this regard, dynamic searchable symmetric encryption (DSSE) is a promising cryptographic tool that enables secure keyword searching over dynamically added or deleted ciphertexts. However, it is still challenging to design a DSSE that provides addition and physical deletion of keywords based on the IDs (we call it ID-based DSSE (IDDSSE)). In other words, the updates need to be done based on the IDs with relevant keywords without leaking the patient's ID to the server. That is required for adding/removing all the information of a data owner when the data is provided/consent is revoked.
Our scheme supports the addition and deletion of genomic records and related information based on ID, which especially complies with the requirements of deleting information of a particular data owner. In our scheme, physical deletion of data is provided to ensure that data owners are in control of the use of their data. There is no scheme in the literature that offers both addition and instant deletion (when the consent is revoked) based on ID. In this regard, we present two new definitions of forward and backward privacy for IDDSSE, and we formally prove that our scheme is secure under non-adaptive attacks.}
\end{abstract}






\maketitle

\section{Introduction}

The rapid advancements in the genomic data generation and availability have influenced associated scientific studies. These massive genomic datasets enable us to understand the connection between many of diseases and genes. 
For the dataset, which is enormous and requires high computing and storage resources, cloud servers are a significant solution. Moreover, to guarantee participants in the study are aware of its objectives and risks, agree to participate willingly with this information, and have the option to revoke their participation subsequently, dynamic informed consent needs to be considered \cite{kaye2015dynamic}.
Dynamic consent provides opportunities for continuing communication between researchers and study participants, which can have a positive impact on research. Legal challenges are emerging in light of the General Data Protection Regulation (GDPR) \cite{gdpr}, which came into effect in the European Union in May 2018 to safeguard personal data. By adhering to dynamic consent, the GDPR protects study participants' safety without restricting biomedical research. Due to its potential to enable participant involvement in research activities across time with the ability to revoke consent at any time, dynamic consent (DC) has attracted interest \cite{budin2017dynamic,prictor2020dynamic,CIC}.


Genomic information is irreversible and can have stigmatising effects on both individuals and their families. Genomic security and privacy are crucial and must be considered since test results, and genetic data are sensitive. Failing to implement privacy and security precautions while storing sensitive genetic information on a public cloud platform leads to privacy and security problems~\cite{erlich2014redefining,erlich2014routes}.
We assume the data server is in the cloud in our model due to a large amount of genomic data. So, the primary goal of our work is to securely outsource genetic data and perform searches on this data while ensuring privacy protection. In our context, each individual piece of genomic information (including Single nucleotide polymorphisms (SNPs) and phenotype data) belonging to a data owner is treated as a separate keyword associated with their unique identifier, ID. This allows for conducting searches on the various pieces of genomic information as distinct keywords, without revealing the actual data or compromising privacy.
As a result, the cloud cannot infer any information beyond what is permitted from the uploaded data and the conducted query. We maintain the feature of consent consideration and revocation in our model. 


\makeatletter
\newcommand{\ssymbol}[1]{^{\@fnsymbol{#1}}}
\makeatother

\begin{table*}[hb]
\caption{Existing dynamic searchable symmetric encryption schemes comparison}
\label{table:first comparison}
\centering
\resizebox{14cm}{!}{%
\begin{tabular}{ccccccccc}
\toprule
\multirow{2}{*}{Scheme}&
\multicolumn{4}{c}{Deletion}&\multirow{2}{*}{}&\multicolumn{2}{c}{Privacy}&\multirow{2}{*}{Comm. cost$^\S$}\\
\cline{2-5}\cline{7-8} &Approach based on&Type&Instant&Non-interactive$\ssymbol{5}$&&FP/BP&ID&\\
\midrule

\cite{sun2018practical} &  w & Logical &\xmark&\xmark &&FP/BP&\xmark& $\mathcal{O}(x)$ \\

\cite{sun2021practical} &  w & Logical &\xmark&\xmark &&FP/BP&\xmark& $\mathcal{O}(x)$ \\

\cite{stefanov2013practical} &  w & Logical &\xmark&\xmark &&FP&\xmark& $\mathcal{O}(x\log(rx))$ \\

\cite{xu2017dynamic}  &  ID &Physical & \xmark &\cmark&&-$\ssymbol{2}$&\xmark& $\mathcal{O}(1)$ \\

\cite{chen2021bestie} &  w & Physical &\xmark &\xmark&&FP/BP&\hspace{5pt}\cmark$\ssymbol{1}$& $\mathcal{O}(x)$ \\

ACE &  ID & Physical  &\cmark&\cmark &&IDFP/IDBP$\ssymbol{3}$&\cmark & $\mathcal{O}(1)$\\
\bottomrule

\multicolumn{9}{p{16cm}}{Notations: FP: Forward Privacy; BP: Backward Privacy; $x$: Number of keywords of an ID; $r$: Number of records (IDs) in DB; $^\S$: Communication cost is compared for the deletion phase, when the information of an ID needs to be deleted; $\ssymbol{5}$: When the data of an ID is deleted; $\ssymbol{2}$: FP/BP have not been discussed in this paper, it was a concurrent work with Bost et al. \cite{bost2017forward} in which they proposed the formal definitions of BP (based on their defined leakages and the formal definitions of FP and BP, this scheme does not provide FP/BP); $\ssymbol{3}$: Please refer to section \ref{sec:security analysis} for the definitions and more details; $\ssymbol{1}$: Their leakage model does not formalize this privacy.}\\
\end{tabular}
}
\end{table*}


A cryptographic technique that enables searching over encrypted data is known as searchable encryption. Dynamic searchable symmetric encryption (DSSE) is a useful technique for protecting user data stored in the cloud that permits the updating of the encrypted database while retaining searchability. However, additional information is revealed during update procedures, which attackers may exploit \cite{cash2015leakage,blackstone2019revisiting,zhang2016all}. DSSE schemes are expected to uphold two new security concepts, forward privacy and backward privacy, which are introduced by Stefanov et al. \cite{stefanov2013practical}. Bost \cite{bost2016ovarphiovarsigma} and Bost et al. \cite{bost2017forward} provided the formal definitions of forward and backward privacy, respectively.
Nevertheless, most existing forward and backward private DSSE schemes are defined to update the database based on a pair of keyword and ID, meaning an update happens for a particular keyword that a data owner with an ID has (keyword can be a single word, a phrase, or any identifiable piece of information of a data owner with identifier ID).

In addition, there are other key requirements for genome searches that cannot be fully satisfied by existing encrypted search schemes, including compliance with dynamic consent and providing instant non-interactive real deletion of data while offering a practical encrypted search mechanism. The following requirements highlight the actual problems faced in achieving efficient and privacy-preserving genome searches.
To ensure compliance with dynamic consent, it is essential to have the capability to remove all data related to a specific ID from the server when a data owner revokes their consent. Existing encrypted search schemes often lack the ability to perform physical deletion of data, making it difficult to comply with the data owner's right to have their data erased and no longer searched (or even processed) after consent revocation.
Moreover, encrypted search schemes should comply with the requirements outlined in the General Data Protection Regulation (GDPR), which \emph{grants individuals the right to have their personal data erased and no longer processed when the data are no longer necessary for the purposes for which they were collected or processed, and the organisation must stop the processing of individual's data and (delete them) as soon as an individual withdraw their consent (you have the right to have your data erased, without undue delay, by the data controller)} \cite{gdpr,gdprerasure}. 

Therefore, the ability to achieve instant non-interactive real deletion is also crucial for consent revocation. It enables removing data from the server instantly, when the consent is revoked, which ensures the individual's right to have their data erased without delay (the right to erasure). Removing data also happens without relying on the interactive client's involvement, which facilitates the management of large-scale datasets. A delay in removing data exposes it to potential unauthorized access or misuse, increasing the risk of data breaches, unauthorized disclosures, and other privacy breaches.
Instant non-interactive real deletion aligns with this GDPR stipulation, enabling encrypted search schemes to adhere to privacy regulations.


DSSE has been investigated to secure data stored on the cloud server, and for updating a pair of keyword and ID, a token needs to be sent to the server~\cite{ghareh2018new,sun2021practical,zuo2019dynamic}. For updating all the keywords of an ID, all the update tokens need to be generated and sent to the server, which incurs a high communication cost for an ID with large number of keywords. For instance, there are discussions on $\Sigma$o$\varphi$o$\varsigma$ protocol presented in \cite{bost2016ovarphiovarsigma} and the construction in \cite{stefanov2013practical} about supporting deletion of data of an ID. $\Sigma$o$\varphi$o$\varsigma$ \cite{bost2016ovarphiovarsigma} needs a token for each pair of keyword and ID, and \cite{stefanov2013practical} rebuilds the data structure for each keyword the ID has. Moreover, the search complexity in \cite{stefanov2013practical} is more than the number of matched IDs for a keyword. None of them supports physical deletion of data, and they reveal the ID as a leakage in their update phase. Authors of \cite{chen2021bestie} propose a construction named Bestie, which supports real deletion. However, the deletion is for a pair of keyword and ID, and happens at the time of the search on that particular keyword. This means for deleting the information of one ID, different tokens for different keywords need to be generated and sent to the server (this can be viewed as a batch deletion operation). Furthermore, the current system retains the data on the server until a search is conducted using a specific keyword (This can result in a significant delay, sometimes spanning years, or in certain cases, the search may not occur at all). However, this practice is not acceptable, especially in cases where a data owner with a specific ID revokes their consent and explicitly requests the removal of their data from the server (the right to erasure). It is crucial that the data is promptly deleted upon consent revocation, rather than being retained until a search is initiated. Other DSSE schemes such as \cite{zuo2021searchable,kasra2022range,sun2018practical,sun2021practical} presented in the literature also support update based on a keyword and ID pair, that is not physically deleting all the information of an ID in the deletion phase.

Moreover, Table \ref{table:first comparison} details an overview of DSSE schemes to show the behaviour of the schemes in deleting ID information, privacy considerations, and the communication cost of deleting an ID. In more detail, the comparison in the deletion phase shows whether it can happen based on an identifier ID or a keyword w, physically or logically and instantly deleted. Logical deletion means keeping the deleted data on the server, but identifying the deleted entries when a query is performed and not including them in the result set. However, physical deletion requires removing the data from the server. Instant deletion means removing the data when deletion is requested and not keeping it for later phases. The schemes, e.g., \cite{chen2021bestie,stefanov2013practical} where the data is kept on the server and is deleted at other times (that can take a while because a search on the keyword needs to happen for the deletion to be completed) are not ideal for providing consent revocation because once consent is revoked, the user expects the relevant data to be deleted immediately. The scheme proposed in \cite{xu2017dynamic} also keeps part of the data and remove it at later stages when a search is performed. Non-interactive deletion when all the keywords of an ID needs to be removed is provided when one deletion token based on ID is generated. Moreover, forward and backward privacy considerations and ID privacy have been considered for comparison. ID privacy relates to the fact that the identifier of the patients/participants needs to be kept private and not revealed to the server at any time.
Ideally, the system should be able to generate a single update token to minimize the communication cost and be able to update all the keywords of an ID on the server, remove the data physically and instantly when the related consent is revoked. The other desirable goal is to provide privacy for the data and the identifiers, IDs. However, there is no existing scheme to achieve/satisfy all of the mentioned points.

Hence, the contributions of this paper are as follows:
\begin{itemize}[noitemsep,topsep=0pt,leftmargin=10pt]
\remove{ 
\item The deletion happens based on ID, which means one token is needed to be sent to the server to remove the corresponding entries of that ID.
    
    \item Only one deletion token (in comparison with generating a token for each keyword of the ID that needs to be removed when the associated consent is revoked) incurs lower communication costs in our scheme.
    
    \item ACE provides real deletion of data. When the consent is revoked, and the server gets the deletion token, it removes the corresponding entries physically. 
    \item Since our structure supports search based on a keyword and deletion based on an ID, the previous notions of privacy do not fit into our structure. Hence, we define two new notions of privacy when updating based on an ID. We formally defined a hard problem (Decisional ACE or D-ACE) and proved its hardness (reduced to DDH problem). Then, we formally proved ACE is secure under non-adaptive attacks.
    
    \item The proposed scheme has been evaluated on genomic data to show its applicability and performance (in terms of the update and search computation and communication costs, storage overhead) of it. However, it is worth mentioning that the proposed construction is applicable as an ID-based DSSE in the applications where the update based on ID is desired.

}

\item We propose a new construction named ACE that leverages two data structures to support search based on keywords and addition/deletion based on ID. The deletion happens based on ID, which means only one token is needed to be sent to the server to remove the corresponding entries of that ID.
Compared with generating a token for each keyword of the ID that needs to be removed when the associated consent is revoked, ACE incurs lower communication costs for performing a delete operation that takes place in a non-interactive way.

\item Our proposed construction, ACE, provides instant real deletion of data. When the consent is revoked, and the server gets the deletion token, it removes the corresponding entries physically, not just logically. Furthermore, in contrast to other schemes that wait for a search to happen on each keyword to be deleted (which might take years for a particular keyword of an individual), ACE removes data when the consent is revoked, without undue delay. The ability to achieve instant non-interactive real deletion is crucial for data management in encrypted search schemes, and it complies with the requirements outlined in the General Data Protection Regulation (GDPR), enabling ACE to adhere to the privacy regulations.

\item Since our structure enables search based on a keyword and deletion based on an ID, the existing notions of forward and backward privacy, which were defined for mechanisms with search based on a keyword and update based on a keyword and ID pair, are not directly applicable to our structure. Hence, we define two new notions of forward (resp. backward) privacy called IDFP (resp. IDBP), to capture privacy for dynamic SSE with updates based on an ID. Then, we prove that ACE achieves privacy under non-adaptive attacks in the sense of our IDFP (resp. IDBP) notion, assuming the hardness of the Decisional Diffie-Hellman (DDH) problem. Our proof makes use of an intermediate computational problem called Decisional-ACE (D-ACE) which we introduce to aid our analysis, and we prove that the hardness of D-ACE follows from the hardness of DDH.

\item We provide implementation result to evaluate the applicability and performance of our ACE DSSE on genomic data sets, in terms of update and search computation costs, communication costs and storage. We show that ACE provides all the above-mentioned features with high performance. Although designed for genomic data applications, our ACE protocol can also be applied as an ID-based DSSE in other applications where update operations based on ID are required.

\end{itemize}


We acknowledge that ACE has been specifically designed to meet the requirements of efficient privacy-preserving search on encrypted genomic data, including Single nucleotide polymorphisms (SNPs) and phenotype data, while also addressing the need for instant real deletion of data upon consent revocation. It is important to note that this construction can also be applied to other applications that demand search functionality over encrypted data with instant real deletion. Therefore, the contributions of ACE extend beyond genomic data and may be of independent interest in various domains.


\subsection{Related works}
Song et al. \cite{song2000practical} introduced the symmetric key encryption to solve the issue of keyword search across encrypted data, that is known as searchable symmetric encryption (SSE). However, the search time of it is linear to the number of keyword/identifier pairs. Later, Goh \cite{goh2003secure} presented a secure indexing technique, in which the search time is linear with the number of files, to enhance the search efficiency. To further improve the search efficiency, Curtmola et al. \cite{curtmola2006searchable} provided a sublinear search time SSE by using inverted index data structure. Moreover, they also formalized the SSE security model (i.e., Real vs. Ideal), which has been adopted in the subsequent research. Later, many SSE schemes with various enhancements were introduced \cite{rangequery,faber2015rich,cash2013highly,kermanshahi2019multi}.
SSE has also been studied to provide privacy-preserving query execution over genomic databases \cite{jafarbeiki2021privgendb,jafarbeiki2021non,jafarbeiki2022pressgendb}. However, these schemes are not dynamic.

To address the need for updating in searchable symmetric encryption (SSE), dynamic SSE (DSSE) schemes have been proposed \cite{kamara2012dynamic,cash2014dynamic}. However, these approaches can inadvertently leak additional information during updates, which can be exploited by adversaries to compromise data privacy. Alternatively, there are schemes such as \cite{naveed2014dynamic}, where the server functions solely as a transmission and storage entity, resulting in reduced information leakage. However, this approach requires multiple rounds of interaction between the client and server and does not provide instant real deletion of data.
In order to mitigate the extra information leakage in SSE, forward and backward privacy are presented informally by Stefanov et al. \cite{stefanov2013practical}. Bost \cite{bost2016ovarphiovarsigma} has formally defined forward privacy, and the formal backward privacy (Type-I, Type-II, and Type-III) is defined by Bost et al. \cite{bost2017forward}. 
In recent years also, different DSSE schemes with varying features of update and privacy have been proposed in the literature \cite{zuo2021searchable,sun2021practical,sun2018practical}.
However, these mentioned DSSE schemes support updating a pair of keyword and ID. To delete the data of one particular ID, different tokens for the keywords are generated and then sent to the server.

There are also some recent schemes to provide privacy and security of genomic data when queries are performed on this type of dataset, including \cite{jafarbeiki2021privgendb,jafarbeiki2022pressgendb,jafarbeiki2021non} that utilised searchable encryption. However, they have not considered dynamic consent in their schemes.


\subsection{Organization}
The subsequent sections of this paper are as follows. Section \ref{section2} gives the necessary background and preliminaries. Section \ref{section:system model} defines the system model and threat model. In Section \ref{sec:ACE}, our proposed construction is presented in detail with the designed algorithms. Section \ref{sec:security analysis} gives the security analysis of our proposed scheme. The analytical performance comparison and the evaluation results are given in Sections \ref{sec:analytical} and \ref{sec:experiments}, respectively. Finally, Section \ref{conclusion} concludes the work.

\section{Preliminaries}\label{section2}
In this section, the required preliminaries are provided. As general preliminaries, we say an algorithm A is efficient if A runs in probabilistic polynomial time. We say a function f($\lambda$) is negligible, denoted negl($\lambda$), if for every constant $c >0$, there exists $\lambda_0$ such that f($\lambda$)$<1/\lambda^c$ for all $\lambda >  \lambda_0$.

\subsection{Genomic data representation}
An organism's whole genetic information is included in its genome. Double-stranded deoxyribonucleic acid (DNA) molecules, that are made up of two long complementary polymer chains, and are used to encode the genome in humans and many other species. Adenine, Cytosine, Guanine, and Thymine are the four basic units known as nucleotides, and they are represented by the letters A, C, G, and T. In the human genome, there are about 3 billion such letters (base pairs). 
Single nucleotide polymorphisms (SNPs) are variations in the genome when more than one base (A, T, C, or G) is identified in a population. 
Most SNPs are biallelic, with just two possible variants (\emph{alleles}) found. An individual's \emph{genotype} is the set of particular alleles they carry.
SNPs make up a significant part of the genetic variation underlying a number of human traits, including height and susceptibility to disease (also known as \emph{phenotype}).

\subsection{Symmetric Key Encryption}
A symmetric key encryption (SE) consists of the following polynomial-time algorithms SE $=(\mathrm{SE} \cdot$ Enc, SE·Dec$)$:
\begin{itemize}[noitemsep,topsep=0pt,leftmargin=10pt]
    \item $ct \leftarrow \operatorname{SE} \cdot \operatorname{Enc}(k, m)$: On input a secret key $k \in \mathcal{K}$ and a message $m \in \mathcal{M}$, it outputs a ciphertext $c t \in \mathcal{C} \mathcal{T}$, where $\mathcal{K}, \mathcal{M}, \mathcal{C} \mathcal{T}$ are the key space, message space and ciphertext space, respectively.
    \item $m \leftarrow \operatorname{SE} \cdot \operatorname{Dec}(k, c t)$: On input the secret key $k$ and the ciphertext $c t$, it outputs the message $m$.
\end{itemize}
Correctness. An SE scheme is perfectly correct if for all message $m \in \mathcal{M}$, secret key $k \in \mathcal{K}$, and $c t \leftarrow \mathbf{S E} \cdot \operatorname{Enc}(k, m)$, it holds that $\operatorname{Pr}[\operatorname{SE} \cdot \operatorname{Dec}(k, c t)]=1$.

\begin{definition}
    
We say an SE is IND-CPA secure if for every probabilistic polynomial time (PPT) adversary $\mathcal{A}$, its advantage
$$
\begin{gathered}
\operatorname{Adv}_{\mathbf{S E}, \mathcal{A}}^{\mathrm{IND}-\operatorname{CPA}}(\lambda)=\mid \operatorname{Pr}\left[\mathcal{A}\left(\operatorname{SE} \cdot \operatorname{Enc}\left(k, m_0\right)\right)=1\right]- \\
\operatorname{Pr}\left[\mathcal{A}\left(\mathcal{A}\left(\operatorname{SE} \cdot \operatorname{Enc}\left(k, m_1\right)\right)=1\right] \mid\right.
\end{gathered}
$$
is negligible, where the secret key $k \in \mathcal{K}$ is kept secret, and $\mathcal{A}$ chooses $m_0, m_1 \in$ $\mathcal{M}$ with equal length. In addition, $\mathcal{A}$ can adaptively issue a polynomial number of encryption queries. For each $m \in \mathcal{M}$, the challenger returns $c t \leftarrow \operatorname{SE} \cdot \operatorname{Enc}(k, m)$.
\end{definition} 

\subsection{Searchable Symmetric Encryption (SSE)}\label{OXT}
Classical data encryption could resolve rising concerns about the security of data that is being outsourced. But in reality, it is more complicated because the cloud server cannot directly search the encrypted data. As a result, the user must download all the data, decrypt them and then do the search. This issue can be resolved owing to Searchable Encryption (SE), which enables the data owner to store data in the cloud in encrypted form while preserving the ability of server to search through encrypted data. Searchable ciphertexts and search tokens are generated by secret key holder in SSE schemes.

\subsection{Dynamic Searchable Symmetric Encryption}

\begin{definition} (DSSE)
Three protocols define a DSSE scheme $\Sigma$ between the client and the server, including $\Sigma.\textup{Setup}$, $\Sigma.\textup{Update}$, and $\Sigma.\textup{Search}$. Their definitions are as follows:

\begin{itemize}[noitemsep,topsep=0pt,leftmargin=10pt]
  \item Protocol $\Sigma.\textup{Setup}(\lambda):$ The client initializes her secret key $K_{\Sigma}$ and an empty state-set $\sigma$ for the security parameter $\lambda$ and sends an empty encrypted database $\mathsf{E D B}$ to the server. The client retains both her key $K_{\Sigma}$ and state-set $\sigma$ private.

  \item Protocol $\Sigma.\textup{Update}\left(K_{\Sigma}, \sigma,\mathsf{ o p,(w, i d) , {E D B}}\right)$: In this protocol, according to parameter op $\in\{\mathsf{a d d,del}\}$, the client adds a new keyword-and-file-identifier entry $\mathsf{(w, i d)}$ to or deletes an existing entry from the server. Given key $K_{\Sigma}$ and state-set $\sigma$, the client sends a new ciphertext of entry $\mathsf{(w, i d)}$ to the server if $\mathsf{op = add}$; otherwise ($\mathsf{op = del}$), she sends a delete token of entry $\mathsf{(w, i d)}$ to the server. The server updates its database $\mathsf{E D B}$ when it receives the aforementioned message. 
  
  \item Protocol $\Sigma.\textup{Search}\left(K_{\Sigma}, \sigma,\mathsf{ w , {E D B}}\right)$ : Given key $K_{\Sigma}$ and state-set $\sigma$, the client sends a search trapdoor of keyword $\mathsf{w}$ to the server. The server performs the search on the keyword over $\mathsf{E D B}$ and returns all valid file identifiers to the client.

\end{itemize}
\end{definition}

To satisfy DSSE correctness, a DSSE scheme has to always locate all valid file identifiers.

In regards to the security of DSSE, a common approach is to define the indistinguishability between a real game and an ideal game of DSSE. The adversary can issue Update and Search queries in both games. In the real game, all keyword-and-file-identifier entries and secret keys are real, and both protocols $\Sigma.\textup{Update}$ and $\Sigma.\textup{Search}$ are correctly implemented. In the ideal game, the responses to all queries of the adversary are simulated by a simulator that only uses leakage functions. We claim that DSSE is secure if a simulator can simulate an ideal game that is indistinguishable from the real game.

\begin{definition}
(Adaptive Security of DSSE). Given leakage functions $\mathcal{L}=$ $\left(\mathcal{L}^{\textup{Setup}}, \mathcal{L}^{\textup{Update}}, \mathcal{L}^{\textup{Search}}\right)$, a DSSE scheme $\Sigma$ is called $\mathcal{L}$-adaptively secure if for any sufficiently large security parameter $\lambda$ and adversary $\mathcal{A}$, there exists an efficient simulator $\mathcal{S}=(\mathcal{S}$.Setup, $\mathcal{S}$.Update, $\mathcal{S}$.Search $)$ for which $\mid \operatorname{Pr}\left[\operatorname{Real}_{\mathcal{A}}^{\Sigma}(\lambda)=\right.$ $1]-\operatorname{Pr}\left[\right.$Ideal$\left._{\mathcal{A}, \mathcal{S}, \mathcal{L}}^{\Sigma}(\lambda)=1\right] \mid$ is negligible in $\lambda$, where games Real $_{\mathcal{A}}^{\Sigma}(\lambda)$ and Ideal ${ }_{\mathcal{A}, \mathcal{S}, \mathcal{L}}^{\Sigma}(\lambda)$ are defined as below:

\begin{itemize}[noitemsep,topsep=0pt,leftmargin=10pt]
  \item Real$_{\mathcal{A}}^{\Sigma}(\lambda)$ : The real game represents the DSSE protocols. Adversary $\mathcal{A}$ can adaptively issue the queries of Update and Search with inputs $\mathsf{(o p,(w, i d))}$ and $\mathsf{w}$, respectively, and then, observe the real transcripts that are generated by the DSSE protocol. In the end, adversary $\mathcal{A}$ outputs a bit.

  \item Ideal$_{\mathcal{A}, \mathcal{S}, \mathcal{L}}^{\Sigma}(\lambda):$ Simulator $\mathcal{S}$ simulates all transcripts. Adversary $\mathcal{A}$ can issue the same queries as in the real game. The $\mathcal{S}$ takes leakage functions $\mathcal{L}$ as input and simulates the corresponding transcripts. In the end, adversary $\mathcal{A}$ outputs a bit.

\end{itemize}
\end{definition}

Let $\mathsf{Q}$ be a list of all queries (Update and Search), and each entry in $\mathsf{Q}$ has the form of $\mathsf{(u, o p,(w, i d))}$ or $\mathsf{(u, w)}$ for the Update and Search, respectively, where $\mathsf{u}$ represents the time of performing a query. Given a keyword $\mathsf{w}$, let function $\operatorname{sp}(\mathsf{w})$ return all the timestamps of the Search queries on keyword $\mathsf{w}$, and function $\operatorname{TimeDB}(\mathsf{w})$ return the undeleted file identifiers of keyword $\mathsf{w}$ and the history timestamps for adding these files, and function $\operatorname{DelHist}(\mathsf{w})$ return the history timestamps of all paired Add and Delete operations about keyword $\mathsf{w}$. Below are the formal definitions of the aforementioned three functions.

\begin{center}
$
\begin{gathered}
\operatorname{sp}(\mathsf{w})=\{\mathsf{u \mid(u, w) \in Q}\} \\
\operatorname{TimeDB}(\mathsf{w})=\{\mathsf{(u, i d) \mid(u,\text{ add},(w, i d)) \in Q} \\\text{ and } \forall \mathsf{u^{\prime},\left(u^{\prime},  { del},(w, i d)\right) \notin Q}\} \\
\operatorname{DelHist}(\mathsf{w})=\left\{\mathsf{\left(u^{a d d}, u^{ {del }}\right) }\mid \exists \mathsf{i d,\left(u^{a d d},{ add},(w, i d)\right) \in Q}\right. \\
\text { and } \left.\left(\mathsf{u^{\text {del }},  { del },(w, i d)}\right) \in \mathsf{Q}\right\}
\end{gathered}
$
\end{center}

With the above functions, forward and Type-III-backward privacy are defined in Appendix \ref{background}. We introduce our new backward privacy notion IDBP, suitable for SSE with ID-based updates like our ACE construction, in Section \ref{sec:security analysis}.


\subsection{Pseudorandom Function (PRF)}
To encrypt search queries and tokens deterministically in our architecture, we employ PRFs. 
A PRF \cite{prf} is a set of effective functions, where no efficient algorithm can distinguish between a randomly chosen function from the PRF family and a random oracle (a function whose outputs are fixed entirely at random), with a significant advantage. Pseudorandom functions are fundamental tools in the cryptographic primitives construction, and are defined as follows:\par
Let $X$ and $Y$ be sets, $F \colon \{0, 1\} ^{\lambda} \times X \rightarrow Y$ be a function, s$\stackrel{\$}{\leftarrow}$S be the operation of allocating to s a randomly selected element from S, $\operatorname{Fun}(X, Y)$ represent the set of all functions from $X$ to $Y$, $\lambda$ represent the security parameter for PRF, and
$\mathsf{negl}(\lambda)$ denotes a negligible function.
We say that $F$ is a pseudorandom function (PRF) if for all efficient adversaries $\mathcal{A}$, $\operatorname{Adv}_{F, \mathcal{A}}^{\mathrm{prf}}(\lambda)=\operatorname{Pr}[\mathcal{A}^{F(K, \cdot)}(1^{\lambda})=1]-\operatorname{Pr}[\mathcal{A}^{f(\cdot)}(1^{\lambda})=1]\leq \mathsf{negl}(\lambda)$, where the probability is over the randomness of $\mathcal{A}$, $K\stackrel{\$}{\leftarrow}\{0,1\}^{\lambda}, \text { and } f \stackrel{\$}{\leftarrow} \operatorname{Fun}(X, Y)$.


\subsection{Trapdoor Permutations}
A trapdoor permutation $\pi$ is a permutation over a set $\mathcal{D}$ such that $\pi$ can be easily evaluated using a public key (PK), but the efficient computation of the inverse, $\pi^{-1}$, requires the use of a secret key (SK).

More formally, $\pi$ is a trapdoor permutation with the key generation algorithm KeyGen if for every efficient adversary $\mathcal{A}$
$$
\operatorname{Adv}^{\mathrm{ow}}_{\pi, \mathcal{A}}(\lambda) \leq \operatorname{negl}(\lambda)
$$
where
\begin{center}
$
\operatorname{Adv}^{\mathrm{ow}}_{\pi, \mathcal{A}}(\lambda)=\operatorname{Pr}[y \stackrel{\$}{\leftarrow} \mathcal{M},(\mathrm{SK}, \mathrm{PK}) \leftarrow \operatorname{KeyGen}\left(1^\lambda\right),$ $x \leftarrow \mathcal{A}\left(1^\lambda, \mathrm{PK}, y\right): \pi_{\mathrm{PK}}(x)=y]
$
\end{center}
$(\pi$ is one-way) while for every $x \in \mathcal{D}$
$$
\pi_{\mathrm{PK}}\left(\pi_{\mathrm{SK}}^{-1}(x)\right)=x \text { and } \pi_{\mathrm{SK}}^{-1}\left(\pi_{\mathrm{PK}}(x)\right)=x
$$
and $\pi_{\mathrm{PK}}(.)$ and $\pi_{\mathrm{SK}}^{-1}(.)$ is computed in polynomial time. 


\section{System model}\label{section:system model}
\subsection{System model overview}
The proposed model is made up of several components (presented in Figure \ref{fig:system design overview}), including data owner, data provider (trustee), data server (genomic sequence data database), and users (analysts or clinicians). Below is a discussion of their roles:\par
\begin{figure}[htp]
\centering
\includegraphics[scale=0.5]{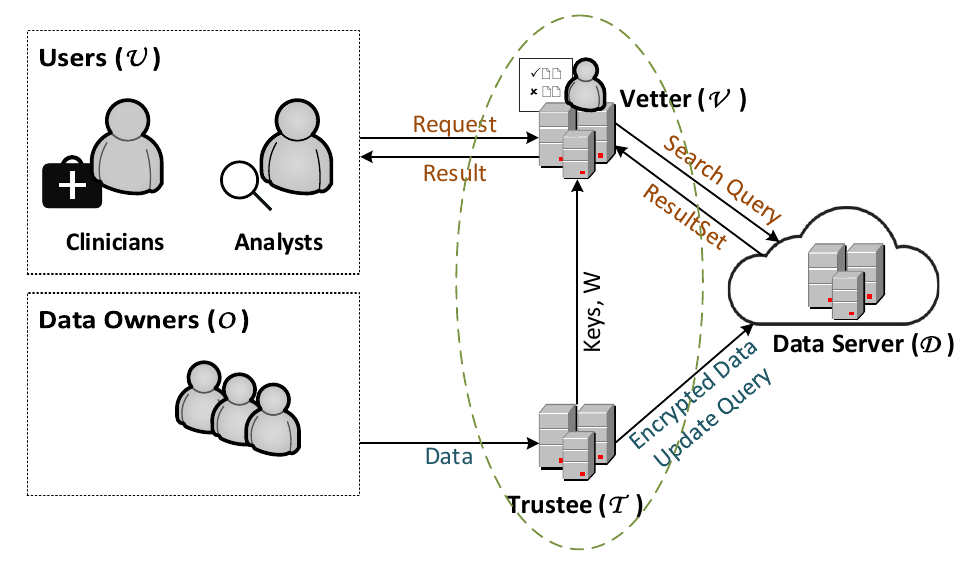}
\caption{System design overview of ACE}
\label{fig:system design overview}
\end{figure}
\textbf{Data owner ($\DO$):}A person whose data is collected is called a data owner. When a data owner attends a medical facility, such as a gene trustee, as a patient or a study participant, her data is taken and recorded while she gives the trustee consent to utilise her genetic data for further studies or treatments. By notifying the trustee, the data owner can subsequently revoke the consent.\par
\textbf{Data provider or Trustee ($\Trustee$):} In our model, a medical institution, like gene trustees, serves as a data provider. 
$\Trustee$ keeps a list of collected genomic data with consent related to them. 
We assume that the data provider is trustworthy. The main responsibilities of this entity are: encoding sequences of genomic data, encrypting the encoded sequences, and managing the cryptographic keys. Moreover, $\Trustee$ is able to insert new genomic data when new data owners provide their samples, and is responsible for removing the genomic data of data owners who revoke their consent.

\textbf{Vetter ($\Vetter$):} There is another Trusted entity that is presented in Figure \ref{fig:system design overview} as a separate entity that also can be combined with the data provider. It receives the keys from the trustee for the search phase, and receives the queries from users and also generates search tokens.
\par
\textbf{Users ($\User$):}
Users send the detailed queries to the trusted entity and wait for the result of the query execution.\par
\textbf{Data Server ($\Server$):} The data server records sequences of genomic data. The $\Server$ executes the encrypted queries on encrypted data and sends back the result. It also stores the newly inserted encrypted data from $\Trustee$ and deletes the requested data based on received update queries from the $\Trustee$.

\subsection{Threat model}
The Data Server ($\mathcal{D}$) should not be able to learn anything regarding the shared genomic data or the unencrypted results of the query that the analysts or clinicians run. This is our ideal security goal.
The Data Server is honest-but-curious (semi-honest) adversary. This proves that $\mathcal{D}$ correctly adheres to the protocol and has no intention of acting intentionally in order to obtain the wrong outcome. However, $\mathcal{D}$ may attempt to obtain additional information than what is anticipated to be obtained during or after the execution of the protocol. We take the trustee to be a trusted entity. Users (Analysts or Clinicians) can be unauthorised, thus they will be authorised by the trustee, that is a trusted entity checking the validity of the query. Finally, we assume that $\mathcal{D}$ and $\User$ do not collude with each other. 
The discussion on the security model and analysis are given in Section \ref{sec:security analysis}.

\section{ACE construction}\label{sec:ACE}

To construct $\ACEPrivGenDB$, we considered the following main ideas. 

 To achieve high search performance, our approach creates searchable ciphertexts in a counter-based manner. By traversing all valid counter values, the counter-based approach enables the server to locate all matching ciphertexts for a keyword.
 This way, the server is able to compute the indices using the counter and decrement it to find the next index. The resulting search complexity is sub-linear with regard to the total number of ciphertexts. This is because the server traverse these computed indices to find matched IDs instead of going through all the indices.
 By considering ${\mathsf{S T}}$ as the parameter that helps in counter-based design, when a search on w1 happens, the server would be able to generate all the ${\mathsf{S T}_{c_1+1}}$ and then ${\mathsf{S T}_{c_1}}$ by using the received ${\mathsf{S T}_{c_1+2}}$ and a token tk=$\mathsf{g^{tag_{w_1}}}$. Therefore, it is able to find the related entries by computing the exact indices using ${\mathsf{S T}}$s and token, tk.
%

 To achieve physical deletion based on ID (when a particular data owner decides to revoke her consent) while ensuring minimal information leakage, we store a set of deltas ($\Delta$) for each ID and issue a token that can be used to generate all the indices related to the ID that the data owner expects to delete. This way, one token for deletion is generated and there is no need for a high communication cost of generating and sending all tokens of all keywords (for an ID) for deletion. 
 The deletion token is based on an $\DOID$, $r_\DOID$ that extracts the deltas of the $\DOID$ and lets the server compute the indices in the ISet using deltas and $\mathsf{tag}_\DOID$. This way, the server can find all the entries in FSet and ISet related to that particular ID to really delete the corresponding ciphertexts. That is why there is no need for sending different tokens to delete all the relative entries of an ID. 

 To achieve ID-based forward privacy (IDFP, defined in section \ref{sec:security analysis}), ACE uses trapdoor permutation ($\mathsf{ST}$) and does not let new insertions to be related to the previous search tokens after insertion. To achieve ID-based backward privacy (IDBP, defined in section \ref{sec:security analysis}), it encrypts all the IDs such that the server learns nothing about the deleted IDs. Since it supports real deletion and the IDs with revoked consents are deleted in the scheme, no deleted ID will be returned whenever a corresponding search query is executed.


\subsection{Notations}
Frequently used notations in this paper are listed in Table \ref{table:notations}.

\begin{table}[htbp]
\caption{Notations}\label{table:notations}
\centering
\begin{tabular}{>{\centering\arraybackslash} m{0.2\linewidth}  m{0.6\linewidth}}
\hline
\toprule
\textbf{Notation} & \textbf{Description} \\
\midrule
$\DOID$& Data Owner's unique $\mathsf{ID}$\\\hline $\mathsf{ID}'$& Encrypted Data Owner's $\mathsf{ID}$\\\hline $\mathsf{w}$& A keyword\\\hline $\DBGePh$($\mathsf{w}$)& The set of Data Owner IDs that contain that particular $\mathsf{w}$ \\\hline
$\mathbf{W}_{\DOID}$& The set of keywords the data owner (with $\DOID$) has \\\hline
$\DBGePh$& Genomic $\mathsf{D}$ata$\mathsf{B}$ase; a set of $\{\DOID_i,\mathbf{W}_{\DOID_i}\}$\\\hline $\mathsf{E}\DBGePh$& Encrypted Genomic $\mathsf{D}$ata$\mathsf{B}$ase\\
\bottomrule
\end{tabular}
\end{table}

\subsection{Construction}


The detailed description of the algorithms of ACE are as follows:
\newcounter{para}
\newcommand\mypara{\par\refstepcounter{para}\thepara)\space}
\mypara{$\mathsf{Setup(\lambda})$}: This process is presented in Algorithm \ref{alg: Setup}.
The Trustee $\Trustee$ runs this algorithm. 
On input the security parameter $\lambda$, $\Trustee$ executes this algorithm and outputs the empty map and dictionary $\mathsf{EGDB=}$ $\{\mathsf{EGDB1,}$ $\mathsf{EGDB2}\}$, an empty map $\mathbf{W}[\mathsf{w}]$ along with the set of keys, $\mathsf{K}$.
    It selects random keys $\K_S, \K_1$ for PRF $F$ and $\K_T,\K_2$ for PRF $F_p$ and the generator $\mathsf{g}\stackrel{\$}{\leftarrow}\mathbb{G}$. It also generates a set of ($\mathsf{SK,PK}$) for $\pi$ using KeyGen algorithm of the trapdoor permutation. 
The $\mathsf{EGDB1}$ stores deltas (that are used for generating tokens for deletion) for each $\DOID$, and the $\mathsf{EGDB2}$ dictionary contains searchable ciphertexts in a counter-based design with the encrypted $\DOID$s.
The $\mathsf{EGDB}$ is stored on the $\Server$, and the relevant keys (for search and retrieve) and a map $\mathbf{W}$ are passed to the $\Vetter$ to produce search tokens. $\Trustee$ keeps all the keys to itself for update phases.


\begin{algorithm}[hbt!]
\caption{\textbf{$\ACEPrivGenDB.\mathsf{Setup}$}}
        \label{alg: Setup}

 \begin{algorithmic}[1]

\State $\Trustee$ select keys $\K_S, \K_1$ for PRF $F$ and ($\mathsf{SK, PK}$) for $\pi $ and keys $\K_T,~\K_2$ for PRF $F_p$ (with range in $\mathbb{Z}^*_p$) and $k_h$ for keyed hash function H
using security parameter $\lambda$, and $\mathbb{G}$ a group of prime order $p$ and generator $\mathsf{g}$.
\State Initialise empty maps $\mathbf{W}[\mathsf{w}], \mathsf{FSet}$ and empty dictionary $\mathsf{ISet}$

\remove{
 \State Parse $\DBGePh$ as $({\DOID}_i, {\mathsf{w}}_j)$ 

\For {each ${\mathsf{w}}$} 
\State $\mathsf{tag}_{\mathsf{w}} \gets F_p(\K_T,{\mathsf{w}})$; $\K_{\mathsf{w}} \gets F(\K_S,{\mathsf{w}})$.
\State Set a counter $c\gets 1$ and select $ \mathsf{S T_{0}} \stackrel{\$}{\leftarrow} \mathcal{M}$
\For {${\DOID}_i \in \DBGePh({\mathsf{w}})$}
\If {there is no index $r_{\DOID}$ in $\mathsf{FSet}$ for ${\DOID}_i$}
\State Compute index $r_{{\DOID}_i}\gets F(\K_1,{\DOID}_i)$
and a tag $\mathsf{tag}_{{\DOID}_i} \gets F_p(\K_2,{{\DOID}_i})$
\EndIf
\State Compute $\DOID' \gets {E}(\K_{\mathsf{w}},\DOID)$ 

\State $ \mathsf{S T}_{c} \leftarrow \pi_{\mathsf{S K}}^{-1}\left(\mathsf{S T}_{c-1}\right)$

\State $\ell \leftarrow H(k_h,\mathsf{g}^{{\mathsf{S T}_{c}}\cdot {\mathsf{tag}_{\mathsf{w}}}})$ 

\State Append $\DOID'$ to ${\bf \mathsf{ISet}[\ell]}$ and $c\gets c+1$\newline
\State Compute $\Delta \gets \mathsf{g}^{{\mathsf{S T}_{c}}\cdot {\mathsf{tag}_{\mathsf{w}}}/{\mathsf{tag}_{{\DOID}_i}}}$ 


\State Append $\Delta$ into ${{ \mathsf{FSet}}[r_{{\DOID}_i}]}$
\EndFor
\State $\mathbf{W}[\mathsf{w}] \leftarrow\left(\mathsf{S T}_{c}, c\right)$
\EndFor

}

\State $\mathsf{E}\DBGePh1 =\mathsf{FSet}, \mathsf{E}\DBGePh2 = \mathsf{ISet}$.
\State \Return{$\mathsf{EGDB}=(\mathsf{E}\DBGePh1,\mathsf{E}\DBGePh2)$} //Stored on $\Server$; $\mathbf{W}[\mathsf{w}]$
 and {$\K_v=(\K_S, \K_T)$} //Sent to $\Vetter$; $\mathsf{SK},\mathbf{W}[\mathsf{w}],\K_t=(\K_S, \K_T, \K_1, \K_2)$ //Stored on $\Trustee$; $\mathsf{PK,g},p,k_h$ are public.
\end{algorithmic}
\end{algorithm}


\begin{algorithm}[hbt!]
\caption{\textbf{$\ACEPrivGenDB.\mathsf{Update}$}}
        \label{alg: token}
\underline{
Add a set of $\DOID$s with their keywords, $\{\DOID_i,\mathbf{W}_{\DOID_i}\}$  (batch insertion)
}
 \begin{algorithmic}[1]

\State $\Trustee$ Parses the set as $\mathsf{GDB}=({\DOID}_i, {\mathsf{w}}_j)$ of $\DOID$ and keyword pairs and also generates $\mathsf{GDB(w)}$ for all distinct keywords $\mathsf{w}$ in $\mathsf{GDB}$, and updates $\mathsf{E}\DBGePh1 =\mathsf{FSet}, \mathsf{E}\DBGePh2 = \mathsf{ISet}$ as follows
\For {each distinct ${\mathsf{w}}\in \mathsf{GDB}$} 
\State $\mathsf{tag}_{\mathsf{w}} \gets F_p(\K_T,{\mathsf{w}})$; $\K_{\mathsf{w}} \gets F(\K_S,{\mathsf{w}})$.//specific ${\DOID}_i$

\State $\left(\mathsf{S T}_{c}, c\right) \leftarrow \mathbf{W}[\mathsf{w}]$
\If {$\left(\mathsf{S T}_{c}, c\right)=\perp$}
\State $\quad \mathsf{S T_{0}} \stackrel{\$}{\leftarrow} \mathcal{M}, c \leftarrow0$
\EndIf

\For {${\DOID}_i \in \DBGePh({\mathsf{w}})$}
\If {there is no index $r_{\DOID}$ in $\mathsf{FSet}$ for ${\DOID}_i$}
\State Compute index $r_{{\DOID}_i}\gets F(\K_1,{\DOID}_i)$
and a tag $\mathsf{tag}_{{\DOID}_i} \gets F_p(\K_2,{{\DOID}_i})$
\EndIf
\State Compute $\DOID'_i \gets {E}(\K_{\mathsf{w}},\DOID_i)$ 
\State $c\gets c+1$
\State $ \mathsf{S T}_{c} \leftarrow \pi_{\mathsf{S K}}^{-1}\left(\mathsf{S T}_{c-1}\right)$; $ \mathsf{S T'}_{c} \leftarrow \left(\mathsf{S T}_{c}\text{ mod }p\right)$
\State $\ell \leftarrow H(k_h,\mathsf{g}^{{\mathsf{S T'}_{c}}\cdot {\mathsf{tag}_{\mathsf{w}}}})$ 

\State Append $\DOID'_i$ to ${\bf \mathsf{ISet}[\ell]}$ on $\Server$ 

\State Compute $\Delta \gets \mathsf{g}^{{\mathsf{S T'}_{c}}\cdot {\mathsf{tag}_{\mathsf{w}}}/{\mathsf{tag}_{{\DOID}_i}}}$ 

\State Append $\Delta$ into ${{ \mathsf{FSet}}[r_{{\DOID}_i}]}$ on $\Server$

\EndFor

\State $\mathbf{W}[\mathsf{w}] \leftarrow\left(\mathsf{S T}_{c}, c\right)$ //gets updated on $\Trustee$ and $\Vetter$
\EndFor


\end{algorithmic}
\underline{Delete all entries for a particular $\DOID_i$}
 \begin{algorithmic}[1]

\State $\Trustee$ computes $\mathsf{tag}_{{\DOID}_i} \gets F_p(\K_2,{{\DOID}_i})$, $r_{{\DOID}_i}\gets F(\K_1,{\DOID}_i)$ and sends to the $\Server$ 
\newline
$\Server$ performs:
\For {all elements $\Delta_i$ in $\mathsf{FSet[r_{{\DOID}_i}]}$}
\State Compute $\ell \gets H(k_h,{\Delta_i}^{{\mathsf{tag}_{{\DOID}_i}}})$
\State Remove corresponding entry from $\mathsf{ISet}[\ell]$ and $\mathsf{\ell}$
\EndFor
\State Remove entries of $\mathsf{FSet[r_{{\DOID}_i}}]$ and $\mathsf{r_{{\DOID}_i}}$
\end{algorithmic}
\end{algorithm}


\mypara{$\mathsf{Update}(\{\DOID_i,\mathbf{W}_{\DOID_i}\},$ $\mathsf{op=add}, X)$ or $\mathsf{Update}(\DOID,$ $\mathsf{op=del}, X)$, where $X=\{\mathsf{SK},\mathbf{W}[\mathsf{w}],\K_t=(\K_S, \K_T, \K_1, \K_2)$ that are stored on $\Trustee$, $\mathsf{EGDB}\}$:}
Based on the operation, $\mathsf{op}$, needed to be performed, either add or delete an ID with its corresponding keywords ($\DOID,\mathbf{W}_{\DOID}$), different steps take place by Trustee $\Trustee$. In ACE, the term update-add refers to the scenario where the data of several new Data Owners are provided to the Trustee (batch insertion), while update-del refers to the situation where a Data Owner revokes their consent and requests the removal of their data. 

Since we have a batch insertion in ACE, if a set of IDs with relevant keywords need to be added, for all the keywords the relative counter is retrieved from the map $\mathbf{W}$ and if it is empty, a random element for $\mathsf{ST_0}$ is selected.
For every $\mathsf{w}$ in the dataset, a tag and a key for encrypting the $\DOID$ are generated. For all the $\DOID$s that have the keyword $\mathsf{w}$ an index $r_{\DOID}$ and a tag $\mathsf{tag}_{\DOID}$ are generated.
To generate the dictionary which has the counter-based search capability $\mathsf{ST}_{c}$ is used, that acts as a counter. In the pseudo code, $\pi$ is a trapdoor permutation and $\mathsf{ST}_{c}$ can be generated by using the secret key of the trapdoor permutation and $\mathsf{ST}_{c-1}$. Then, an index $\ell$ which is based on counter ($\mathsf{ST}_{c}$ mod $p$) and $\mathsf{w}$ ($\mathsf{tag_w}$) is generated and the relevant encrypted ID is appended to the dictionary with index $\ell$. We use mod $p$ to be able to perform the operation in group $\mathbb{G}$ of prime order $p$. These indices $\ell$ and the corresponding encrypted IDs create the $\mathsf{ISet} $ that is considered as the $\mathsf{EGDB2}$. $\mathsf{EGDB1}$ is a map that stores different deltas, $\Delta$, for a particular ID. In this case, when looking for an ID, the corresponding deltas will be retrieved which are based on counter ($\mathsf{ST}_{c}$), $\mathsf{w}$ ($\mathsf{tag_w}$), ID ($\mathsf{tag_{\DOID}}$). This way, when a search token is sent to the server, it would not be able to calculate the indices in $\mathsf{ISet}$ using deltas and find the correlation of deltas and indices in the $\mathsf{ISet}$.
On the Vetter $\Vetter$ and Trustee $\Trustee$ sides, $\mathbf{W}$ maps every inserted keyword to its current $\mathsf{ST}_{c}$ and to a counter $c$. Every time a new document matching $\mathsf{w}$ is inserted, $\mathbf{W}[\mathsf{w}]$ gets incremented. 
So, $\mathsf{FSet}$ and $\mathsf{ISet}$ are computed and stored on the data server and new $\mathsf{ST}$ and counter $c$ are stored in $\mathbf{W}$.

For deleting an ID when the consent is revoked, a tag for that particular ID is generated by Trustee $\Trustee$ and sent to the Server $\Server$. Accordingly, the $\Server$ retrieve the deltas in $\mathsf{FSet}$ and starts computing the corresponding indices in the $\mathsf{ISet}$ using deltas and the received token. 
After computation and searching for these indices, all the relative entries in $\mathsf{FSet,ISet}$ are removed by the $\Server$.


\remove{

\begin{table*}[ht]
\caption{Example of $\mathsf{FSet}$ and $\mathsf{ISet}$ in $\ACEPrivGenDB$ assuming $\DOID_1$ has keywords $\mathsf{w1,w2}$, $\DOID_2$ has keywords $\mathsf{w1,w3}$, and $\DOID_3$ has keywords $\mathsf{w1,w3}$}\label{table:example}
\centering
\begin{tabular}{lll|ll|}
\cline{1-2} \cline{4-5}
\multicolumn{2}{|c|}{$\mathsf{FSet}$}                                                          &  & \multicolumn{2}{c|}{$\mathsf{ISet}$}                \\ \cline{1-2} \cline{4-5} 
\multicolumn{1}{|l|}{$r_{\DOID_1}$} & \multicolumn{1}{l|}{$\mathsf{g}^{{\mathsf{S T}_{c_1}}\cdot {\mathsf{tag}_{\mathsf{w}_1}}/{\mathsf{tag}_{{\DOID}_1}}}$, $\mathsf{g}^{{\mathsf{S T}_{c_2}}\cdot {\mathsf{tag}_{\mathsf{w}_2}}/{\mathsf{tag}_{{\DOID}_1}}}$}       &  & \multicolumn{1}{l|}{$H(k_h,\mathsf{g}^{{\mathsf{S T}_{c_1}}\cdot {\mathsf{tag}_{\mathsf{w}_1}}})$}   & ${E}(\K_{\mathsf{w_1}},\DOID_1)$ \\ \cline{1-2} \cline{4-5} 
\multicolumn{1}{|l|}{$r_{\DOID_2}$} & \multicolumn{1}{l|}{$\mathsf{g}^{{\mathsf{S T}_{c_1+1}}\cdot {\mathsf{tag}_{\mathsf{w}_1}}/{\mathsf{tag}_{{\DOID}_2}}}$, $\mathsf{g}^{{\mathsf{S T}_{c_3}}\cdot {\mathsf{tag}_{\mathsf{w}_3}}/{\mathsf{tag}_{{\DOID}_2}}}$}     &  & \multicolumn{1}{l|}{$H(k_h,\mathsf{g}^{{\mathsf{S T}_{c_2}}\cdot {\mathsf{tag}_{\mathsf{w}_2}}})$}   & ${E}(\K_{\mathsf{w_2}},\DOID_1)$ \\ \cline{1-2} \cline{4-5} 
\multicolumn{1}{|l|}{$r_{\DOID_3}$} & \multicolumn{1}{l|}{$\mathsf{g}^{{\mathsf{S T}_{c_1+2}}\cdot {\mathsf{tag}_{\mathsf{w}_1}}/{\mathsf{tag}_{{\DOID}_3}}}$,   $\mathsf{g}^{{\mathsf{S T}_{c_3+1}}\cdot {\mathsf{tag}_{\mathsf{w}_3}}/{\mathsf{tag}_{{\DOID}_3}}}$} &  & \multicolumn{1}{l|}{$H(k_h,\mathsf{g}^{{\mathsf{S T}_{c_1+1}}\cdot {\mathsf{tag}_{\mathsf{w}_1}}})$} & ${E}(\K_{\mathsf{w_1}},\DOID_2)$ \\ \cline{1-2} \cline{4-5} 
                          &                                                         &  & \multicolumn{1}{l|}{$H(k_h,\mathsf{g}^{{\mathsf{S T}_{c_3}}\cdot {\mathsf{tag}_{\mathsf{w}_3}}})$}   & ${E}(\K_{\mathsf{w_3}},\DOID_2)$ \\ \cline{4-5} 
                          &                                                         &  & \multicolumn{1}{l|}{$H(k_h,\mathsf{g}^{{\mathsf{S T}_{c_1+2}}\cdot {\mathsf{tag}_{\mathsf{w}_1}}})$} & ${E}(\K_{\mathsf{w_1}},\DOID_3)$ \\ \cline{4-5} 
                          &                                                         &  & \multicolumn{1}{l|}{$H(k_h,\mathsf{g}^{{\mathsf{S T}_{c_3+1}}\cdot {\mathsf{tag}_{\mathsf{w}_3}}})$} & ${E}(\K_{\mathsf{w_3}},\DOID_3)$ \\ \cline{4-5} 
\end{tabular}
\end{table*}


}


\mypara{$\mathsf{Search}(\mathbf{W}[\mathsf{w}],\K_v,\mathsf{w},$ $\mathsf{E}\DBGePh_2)$}: 
The Vetter $\Vetter$ generates a token for the search and also retrieves the corresponding counter and $\mathsf{ST}$ from map $\mathbf{W}$ to send to the $\Server$ for the search process.
The $\Server$ starts computing the indices in the $\mathsf{ISet}$ based on the counter (using trapdoor permutation and its public key) and retrieves the encrypted IDs. The whole process is described in Algorithm \ref{alg:search}. $\Server$ creates an empty set $\mathsf{RSet}$ to put related encrypted IDs ($\DOID')$ matched the query in it. Then, the $\Vetter$ gets the $\mathsf{RSet}$, and generates the key for decrypting the retrieved $\DOID'\in\mathsf{RSet}$ by using $\K_v,\mathsf{w}$ (Dec is the decryption algorithm).


\begin{algorithm}[hbt!]
\caption{\textbf{$\ACEPrivGenDB.\mathsf{Search}$}}
        \label{alg:search}

\begin{algorithmic}[1]

\State $\Vetter$ computes $\mathsf{tag}_{\mathsf{w}} \gets F(\K_T,{\mathsf{w}})$, $\mathsf{tk}\gets \mathsf{g}^{\mathsf{tag}_{\mathsf{w}}}$ and gets $\left(\mathsf{S T}_{c}, c\right) \leftarrow \mathbf{W}[\mathsf{w}]$

\If {$\left(\mathsf{S T}_{c}, c\right)=\perp$}
\State  return $\emptyset$
\EndIf
\State Send $\left(\mathsf{tk}, \mathsf{S T}_{c}, c\right)$ to the $\Server$.\newline
$\Server$ performs the following on $\mathsf{E}\DBGePh2 = \mathsf{ISet}$:
\State $\mathsf{RSet}$ $\gets \{\}$
\For {$i=c$ to 1}
\State $ \ell \leftarrow H(k_h,\mathsf{tk}^{({\mathsf{S T}_{i}}\text{ mod }p)})$
\State $ \DOID' \leftarrow \mathsf{ISet}\left[\ell\right]$//skips this if the entry is removed
\State $\mathsf{RSet}$ $\gets \mathsf{RSet}\cup \DOID'$
\State $\mathsf{S T}_{i-1} \leftarrow \pi_{\mathrm{PK}}\left(\mathsf{S T}_{i}\right)$
\EndFor

\State \Return{$\mathsf{RSet}$} to $\Vetter$\newline
$\Vetter$ performs the following

  \State $\Vetter$ defines $\mathsf{IDSet}$ $\gets \{\}$ and performs the following:
  \State Sets $\K_{\mathsf{w}} \gets F(\K_S, \mathsf{w})$
  \For {each $\DOID' \in \mathsf{RSet}$}
  \State Compute $\DOID \gets Dec(\mathsf{K}_{\mathsf{w}}, \DOID')$
  \State $\mathsf{IDSet} \gets$ $\mathsf{IDSet}\cup \{\DOID\}$
  \EndFor
\State \Return $\mathsf{IDSet}$

\end{algorithmic}

\end{algorithm}

\remove{

\mypara{$\mathsf{Retrieve}(\mathsf{RSet}, \mathsf{w}, \K_v)$}: As shown in Algorithm \ref{alg: retrieve}, the $\Vetter$ generates the key for decrypting the retrieved $\DOID'\in\mathsf{RSet}$ by using $\K_v,\mathsf{w}$ (Dec is the decryption algorithm).

\begin{algorithm}[hbt!]
\caption{\textbf{$\ACEPrivGenDB.\mathsf{Retrieve}$}}
        \label{alg: retrieve}

 \begin{algorithmic}[1]
  \State $\Vetter$ defines $\mathsf{IDSet}$ $\gets \{\}$ and performs the following:
  \State Sets $\K_{\mathsf{w}} \gets F(\K_S, \mathsf{w})$
  \For {each $\DOID' \in \mathsf{RSet}$}
  \State Compute $\DOID \gets Dec(\mathsf{K}_{\mathsf{w}}, \DOID')$
  \State $\mathsf{IDSet} \gets$ $\mathsf{IDSet}\cup \{\DOID\}$
  \EndFor
\State \Return $\mathsf{IDSet}$
\end{algorithmic}
\end{algorithm}

}

\subsection{An Example of Stored Data in ACE}
Table \ref{table:example} shows an example of the stored data on data server. Stored deltas in FSet where $\Delta_{ij}$ ($i$ determines the $\DOID_i$ and $j$ determines $\mathsf{w}_j$) help with the deletion of an ID's data without revealing any relationship between the entries of FSet and ISet before deletion. The deltas are generated using $\mathsf{ST}$ acting as counters (they provide privacy features that are discussed in details in section \ref{sec:security analysis}), $\mathsf{tag_w}$ related to the keyword $\mathsf{w_j}$ and $\mathsf{tag_\DOID}$ related to the $\mathsf{\DOID_i}$. The indices in ISet can be generated using deltas and $\mathsf{tag_\DOID}$. This relationship between entries in FSet and ISet is not computable by server; unless a deletion of $\DOID$ needs to happen. 


\begin{table*}[htb]
\caption{Example of $\mathsf{FSet}$ and $\mathsf{ISet}$ in $\ACEPrivGenDB$ assuming $\DOID_1$ has keywords $\mathsf{w1,w2}$, $\DOID_2$ has keywords $\mathsf{w1,w3}$, and $\DOID_3$ has keywords $\mathsf{w1,w3}$}\label{table:example}
\centering
\begin{tabular}{lll|ll|}
\cline{1-2} \cline{4-5}
\multicolumn{2}{|c|}{$\mathsf{FSet}$}                                                          &  & \multicolumn{2}{c|}{$\mathsf{ISet}$}                \\ \cline{1-2} \cline{4-5} 
\multicolumn{1}{|l|}{$r_{\DOID_1}$} & \multicolumn{1}{l|}{$\Delta_{11}=\mathsf{g}^{{\mathsf{S T}_{c_1}}\cdot {\mathsf{tag}_{\mathsf{w}_1}}/{\mathsf{tag}_{{\DOID}_1}}}$, $\Delta_{12}=\mathsf{g}^{{\mathsf{S T}_{c_2}}\cdot {\mathsf{tag}_{\mathsf{w}_2}}/{\mathsf{tag}_{{\DOID}_1}}}$}       &  & \multicolumn{1}{l|}{$H(k_h,\mathsf{g}^{{\mathsf{S T}_{c_1}}\cdot {\mathsf{tag}_{\mathsf{w}_1}}})$}   & ${E}(\K_{\mathsf{w_1}},\DOID_1)$ \\ \cline{1-2} \cline{4-5} 
\multicolumn{1}{|l|}{$r_{\DOID_2}$} & \multicolumn{1}{l|}{$\Delta_{21}=\mathsf{g}^{{\mathsf{S T}_{c_1+1}}\cdot {\mathsf{tag}_{\mathsf{w}_1}}/{\mathsf{tag}_{{\DOID}_2}}}$, $\Delta_{23}=\mathsf{g}^{{\mathsf{S T}_{c_3}}\cdot {\mathsf{tag}_{\mathsf{w}_3}}/{\mathsf{tag}_{{\DOID}_2}}}$}     &  & \multicolumn{1}{l|}{$H(k_h,\mathsf{g}^{{\mathsf{S T}_{c_2}}\cdot {\mathsf{tag}_{\mathsf{w}_2}}})$}   & ${E}(\K_{\mathsf{w_2}},\DOID_1)$ \\ \cline{1-2} \cline{4-5} 
\multicolumn{1}{|l|}{$r_{\DOID_3}$} & \multicolumn{1}{l|}{$\Delta_{31}=\mathsf{g}^{{\mathsf{S T}_{c_1+2}}\cdot {\mathsf{tag}_{\mathsf{w}_1}}/{\mathsf{tag}_{{\DOID}_3}}}$,   $\Delta_{33}=\mathsf{g}^{{\mathsf{S T}_{c_3+1}}\cdot {\mathsf{tag}_{\mathsf{w}_3}}/{\mathsf{tag}_{{\DOID}_3}}}$} &  & \multicolumn{1}{l|}{$H(k_h,\mathsf{g}^{{\mathsf{S T}_{c_1+1}}\cdot {\mathsf{tag}_{\mathsf{w}_1}}})$} & ${E}(\K_{\mathsf{w_1}},\DOID_2)$ \\ \cline{1-2} \cline{4-5} 
                          &                                                         &  & \multicolumn{1}{l|}{$H(k_h,\mathsf{g}^{{\mathsf{S T}_{c_3}}\cdot {\mathsf{tag}_{\mathsf{w}_3}}})$}   & ${E}(\K_{\mathsf{w_3}},\DOID_2)$ \\ \cline{4-5} 
                          &                                                         &  & \multicolumn{1}{l|}{$H(k_h,\mathsf{g}^{{\mathsf{S T}_{c_1+2}}\cdot {\mathsf{tag}_{\mathsf{w}_1}}})$} & ${E}(\K_{\mathsf{w_1}},\DOID_3)$ \\ \cline{4-5} 
                          &                                                         &  & \multicolumn{1}{l|}{$H(k_h,\mathsf{g}^{{\mathsf{S T}_{c_3+1}}\cdot {\mathsf{tag}_{\mathsf{w}_3}}})$} & ${E}(\K_{\mathsf{w_3}},\DOID_3)$ \\ \cline{4-5} 
\end{tabular}
\end{table*}



\section{Security analysis}\label{sec:security analysis}
%

The real world versus ideal world formalization is used in the SSE scheme's confidentiality definition, and a leakage function that describes the information the protocol leaks to the adversary parametrizes it. The definition makes sure that the scheme only leaks data that is directly inferrable from the leakage function.

More precisely, the security definition of the proposed constructions is formulated by two games; $\operatorname{Real}_{\mathcal{A}}^{\Pi}(\lambda)$ and $\operatorname{Ideal}_{\mathcal{A}, \mathcal{S}}^{\Pi}(\lambda)$. The former is executed using our scheme, whereas the latter is simulated using the leakage of our scheme. The leakage is parameterised by a function $\mathcal{L}=\left(\mathcal{L}^{S t p}, \mathcal{L}^{S r c h}, \mathcal{L}^{U p d t}\right)$, which describes what information is leaked to the adversary $\mathcal{A}$. If an adversary such as $\mathcal{A}$ cannot distinguish these two games, then we can say that there is no leakage beyond what is defined in the leakage function.

To enable us to handle ID-based deletion queries in our security reduction of ACE,  we define a non-adaptive security model where some information about the adversary's queries are defined by the adversary in the beginning of the game using a data structure called query info. We define query-info to be a set of queries defined by adversary in advance. This set includes: IDs to be deleted, keywords of those IDs to be searched before deletion (from this information, a set called $S_{i}\text{ with }\{t_{Srch}<t_{Del}\}$ for each ID$_i$ can be created that includes the keywords of that ID that are searched before being deleted). The update-add queries are not included in query-info if the added IDs are not in the to be deleted list of IDs.\par
\begin{center}
query-info$=\{(\DOID_i,S_i)|\DOID_i\text{ will be deleted and}$\\ $S_i=\text{set of }\mathsf{w}\in\mathbf{W}_{\DOID_i}\text{ with } t_{Srch}<t_{Del}\}$\par
\end{center}
The games can be formally defined as followed;\par
- $\operatorname{Real}_{\mathcal{A}}^{\Sigma}(\lambda)$ : On input a dataset and query-info chosen by the adversary $\mathcal{A}$, it outputs $\mathrm{EGDB}$ by using the real algorithms (Setup, Update-add) to $\mathcal{A}$. The adversary can perform the search and update-del queries in query-info and other search and update-add queries. The game outputs the results generated by running Search and Update to $\mathcal{A}$. Eventually, $\mathcal{A}$ outputs a bit.\par
- $\operatorname{Ideal}_{\mathcal{A}, \mathcal{S}}^{\Sigma}(\lambda)$ : On input a dataset and query-info chosen by $\mathcal{A}$, it uses a simulator $\mathcal{S}\left(\mathcal{L}^{S t p},\mathcal{L}^{U p d t}\right)$ to output $\mathrm{EGDB}$ to the adversary $\mathcal{A}$. Then, it simulates the results for the search query using the leakage function $\mathcal{S}\left(\mathcal{L}^{\text {Srch }}\right)$ and uses $\mathcal{S}\left(\mathcal{L}^{U p d t}\right)$ to simulate the results for update (add or delete) query and uses query-info (that is defined in advance by $\mathcal{A}$) when simulating the results for add queries. Eventually, $\mathcal{A}$ outputs a bit.\par

\begin{definition}

(Security w.r.t. Server). The protocol $\Pi$ is $\mathcal{L}$ semantically secure against non-adaptive attacks if for all $P P T$ adversaries $\mathcal{A}$, there exists a PPT simulator $\mathcal{S}$, such that
$$
\left|\operatorname{Pr}\left[\operatorname{Real}_{\mathcal{A}}^{\Pi}(\lambda)=1\right]-\operatorname{Pr}\left[\operatorname{Ideal}_{\mathcal{A}, \mathcal{S}}^{\Pi}(\lambda)\right]\right| \leq \operatorname{negl}(\lambda)
$$
\end{definition}

The security of our scheme can be proven in the Random Oracle Model (we show the security of this construction when H is modeled as a random oracle).

\subsection{Security Assumptions}
In this section, we define a hard problem, named D-ACE, that facilitates the proof of our theorem. We formally prove that D-ACE is a hard problem. Otherwise, DDH problem can be solved (a reduction from DDH to D-ACE is presented).

\begin{definition}
(Multi-instance DDH problem). Let $\mathbb{G}$ be a cyclic group of prime order $p$, the multi-instance decisional Diffie-Hellman (DDH) problem is to distinguish the ensemble
 $\{(\mathsf{g}, \mathsf{g}^{r_i}, \mathsf{g}^{s_j}, \mathsf{g}^{{r_i} {s_j}})\}_{i,j}$ from $\{(\mathsf{g}, \mathsf{g}^{r_i}, \mathsf{g}^{s_j}, \mathsf{g}^{z_{i,j}})\}_{i,j}$ with independent uniform $z_{i,j}$s, where $i=1,\ldots,m$ and $j = 1,\ldots,n$, for some $m,n$ polynomial in security parameter $\lambda$,
$\mathsf{g} \in \mathbb{G}$ and $r_i, s_j, z_{i,j} \in \mathbb{Z}_p$ are chosen uniformly at random. We say the multi-instance of DDH assumption holds if for all PPT distinguisher $\mathcal{D}$, its advantage $\operatorname{A d v}_{\mathcal{D}, \mathrm{G}}^{D D H}(\lambda)\text{ is equal to: }\\ \mid \operatorname{Pr}[\mathcal{D}(\mathsf{g}, \mathsf{g}^{r_i}, \mathsf{g}^{s_j}, \mathsf{g}^{{r_i} {s_j}})_{i,j}=$ $1]-\operatorname{Pr}[\mathcal{D}(\mathsf{g}, \mathsf{g}^{r_i} \mathsf{g}^{s_j}, \mathsf{g}^{z_{i,j}})_{i,j}=1] \mid \leq \text{negl}(\lambda)$, where $\text{negl}(\lambda)$ is negligible in $\lambda$.\par
Remark: It is well known (by a standard hybrid reduction) that the hardness of multi-instance DDH for m,n=poly($\lambda$) is equivalent to the standard one-instance DDH problem with m=n=1 \cite{boneh2020graduate}.

\end{definition}

\remove{
\begin{definition}
(DDH problem). Let $\mathbb{G}$ be a cyclic group of prime order $p$, the decisional Diffie-Hellman (DDH) problem is to distinguish the ensemble $\{(\mathsf{g}, \mathsf{g}^a, \mathsf{g}^b, \mathsf{g}^{a b})\}$ from $\{(\mathsf{g}, \mathsf{g}^a, \mathsf{g}^b, \mathsf{g}^z)\}$, where $\mathsf{g} \in \mathbb{G}$ and $a, b, z \in \mathbb{Z}_p$ are chosen uniformly at random. We say the DDH assumption holds if for all PPT distinguisher $\mathcal{D}$, its advantage $\operatorname{A d v}_{\mathcal{D}, \mathrm{G}}^{D D H}(\lambda)\text{ is equal to: }\\ \mid \operatorname{Pr}[\mathcal{D}(\mathsf{g}, \mathsf{g}^a, \mathsf{g}^b, \mathsf{g}^{a b})=$ $1]-\operatorname{Pr}[\mathcal{D}(\mathsf{g}, \mathsf{g}^a, \mathsf{g}^b, \mathsf{g}^z)=1] \mid \leq \text{negl}(\lambda)$, where $\text{negl}(\lambda)$ is negligible in $\lambda$
\end{definition}
}


\begin{definition}
(D-ACE problem). Let $\mathbb{G}$ be a cyclic group of prime order $p$, and $\pi$ be a permutation with a KeyGen algorithm that generates a set of key $\mathsf{(PK, SK)}$ for the $\pi$ evaluation, $\lambda$ be the security parameter, $\mathcal{A}$ be the adversary, and consider the game in Algorithm \ref{alg: assumption2} that is played between an adversary $\mathcal{A}$ and a challenger and is parameterized by a bit $v\in \{0, 1\}$.
The adversary's distinguishing advantage is $|\operatorname{Pr}[v=v']-(1/2)|$ and we say that D-ACE assumption holds if for all PPT adversary, its distinguishing advantage $|\operatorname{Pr}[v=v']-(1/2)|\mid \leq \text{negl}(\lambda)$, where negl$(\lambda)$ is negligible in $\lambda$.

\end{definition}

\begin{algorithm}[hbt!]

\caption{\textbf{D-ACE}}
        \label{alg: assumption2}

\begin{algorithmic}[1]
\State $(\mathsf{SK,PK})\gets$KeyGen$(1^{\lambda})$
\State $\mathcal{A}$ picks two scenarios of $S_{0}$(computing the elements of a set), $S_{1}$(random elements in a set) and chooses m and n.
\State $v \stackrel{\$}{\leftarrow} \{0,1\}$

\For {$i=1$ to m} 
\State Randomly pick $a_i \stackrel{\$}{\leftarrow} \{0,1\}^\lambda$
\State Select $ {b'_{0}} \stackrel{\$}{\leftarrow} \mathcal{M}$
\For {$j=1$ to n}
\State Randomly pick $c_j \stackrel{\$}{\leftarrow} \{0,1\}^\lambda$

\State $b'_{j} \leftarrow \pi_{\mathsf{S K}}^{-1}\left(b'_{j-1}\right)$
\State $b_j \gets b'_j$ mod $p$
\State $\mathsf{I}_{ij} \leftarrow \mathsf{g}^{{b_j}\cdot {a_i}}$
\If{$v=0$}
\State Compute $\mathsf{F}_{ij} \gets \mathsf{g}^{{b_j}\cdot {a_i}/{c_j}}$ 
\State else
\State Randomly select $\mathsf{F}_{ij} \stackrel{\$}{\leftarrow} \mathbb{G}$

\EndIf


\EndFor
\EndFor

\State $v'\gets\mathcal{A}(\mathsf{g}, \mathsf{g}^{a_i}, b_j,
\mathsf{F}_{ij})_{i,j}$


\end{algorithmic}
\end{algorithm}

\begin{algorithm}[hbt!]

\caption{\textbf{Reduction from DDH to D-ACE}}
        \label{alg: lemma}

\begin{algorithmic}[1]
\State $(\mathsf{SK,PK})\gets$KeyGen$(1^{\lambda})$

\State Given all 
$\{(\mathsf{g}, \mathsf{g}^{r_i}, \mathsf{g}^{s_j}, U_{i,j})\}_{i,j}$, with $U_{i,j}$ being either real ($\mathsf{g}^{{r_i}{s_j}}$) or random ($\mathsf{g}^{z_{i,j}}$) for independent uniform $z_{i,j}$ in $\mathbb{Z}_p$
, the DDH will run $\mathcal{A}$

\For {$i=1$ to m}
\State Select $ {b'_{0}} \stackrel{\$}{\leftarrow} \mathcal{M}$
\For {$j=1$ to n}

\State $b'_{j} \leftarrow \pi_{\mathsf{S K}}^{-1}\left(b'_{j-1}\right)$

\State $b_j \gets b'_j$ mod $p$
\remove{
\If{$U_{ij}=\mathsf{g}^{a_i/c_j}$} //Real case in DDH
\State DDH alg. computes $U_{ij}^{b_j}$ 
\State Runs $\mathcal{A}$ on $\mathcal{A}(\mathsf{g}, \mathsf{g}^{a_i}, b_j,\mathsf{g}^{a_ib_j},U_{ij}^{b_j}=\mathsf{g}^{a_i\cdot b_j/c_j})$
//Real case in D-ACE
\State $v=0$

\EndIf

\If{$U_{ij}=g^{z_{ij}}$} //Random case in DDH
\State DDH alg. computes $U_{ij}^{b_j}$ 
\State Runs $\mathcal{A}$ on $\mathcal{A}(\mathsf{g}, \mathsf{g}^{a_i}, b_j,\mathsf{g}^{a_ib_j},U_{ij}^{b_j}=\mathsf{g}^{z_{ij}\cdot b_j})$ //Random case in D-ACE
\State $v=1$

\EndIf
}

\EndFor
\EndFor

\State $v'\gets\mathcal{A}(\mathsf{g}, \mathsf{g}^{r_i}, b_j,U_{ij}^{b_j})_{i,j}$

\State $\mathcal{D}$ outputs $v'$

\end{algorithmic}
\end{algorithm}

\begin{lemma}
If there exists an efficient algorithm $\mathcal{A}$ with a non-negligible advantage against D-ACE, then we can construct an efficient algorithm $\mathcal{D}$ with a non-negligible advantage against DDH.
\end{lemma}
\begin{proof}
The reduction algorithm (Algorithm \ref{alg: lemma}) uses the $r_i$, $s_j$ of the DDH input  instance as the $a_i$, $1/c_j$ of the D-ACE instance,  respectively, whereas the $b_j$s of the D-ACE instance are simulated by $\mathcal{D}$ itself exactly as in the D-ACE game. The reduction can be analyzed as follows by considering the two possible cases of inputs to $\mathcal{D}$. If the input to algorithm $\mathcal{D}$ comes from the real DDH distribution i.e., $U_{i,j} = \mathsf{g}^{{r_i}{s_j}}$, then the last input to $\mathcal{A}$ in line 10 is $U_{i,j}^{b_j} = \mathsf{g}^{{r_i}{s_j}{b_j}} = \mathsf{g}^{{a_i} {b_j} / {c_j}}$, exactly as in the D-ACE real game ($v=0$), while if the input to $\mathcal{D}$ comes from the random DDH distribution i.e., $U_{i,j} = \mathsf{g}^{z_{i,j}}$, then the last input to $\mathcal{A}$ in line 10 is $U_{i,j}^{b_j} = \mathsf{g}^{{z_{i,j}}{b_j}}$, which are uniform and independent group elements if $b_j \neq 0$ for all $j$. Therefore, adv($\mathcal{D}$) can differ from adv($\mathcal{A}$) by at most the probability of the event B that one of the $b_j$s $= 0$.



In Algorithm \ref{alg: lemma}, when $z_{ij}$ is uniform, we want $z_{ij}\cdot b_j$ to be uniform. If $b_j$ is invertible in mod $p$, uniform $z_{ij}$ mod $p$ gives uniform $z_{ij}\cdot b_j$ mod $p$.
Based on line 4 in Algorithm \ref{alg: lemma}, $b'_j$ is uniform in $\mathcal{M}$ since $b'_0$ is uniform and $b'_j$ gets mapped through an iterated permutation (line 6). Therefore, we have:
\begin{center}
    $|\operatorname{Pr}[b_j\text{ mod }p=0]|=$\\$|\operatorname{Pr}[\text{a uniform element in }\mathcal{M}\text{ mod }p=0]|\leq (1/p)$

\end{center}

Now, for all $j$, we have: 
  $|\operatorname{Pr}[\exists j|_{1}^{n} ~b_j\text{ mod }p=0]|\leq (n/p)$  

this is negligible in $\lambda$ ($p\geq 2^{2\lambda}$ is large). This means except with probability equals to $n/p$, which is negligible in $\lambda$, all of the $b_j$ are not $0$ and uniform $z_{ij}$ maps to uniform $(b_j\cdot z_{ij})$ and the reduction works as in the given Algorithm \ref{alg: lemma}.
  \end{proof}



\subsection{Leakages}
Let list $\mathsf{Q}$ be a set of all Update and Search queries, where each entry in list $\mathsf{Q}$ has the form of ($\mathsf{u,add,(\DOID_1,\DOID_2,\ldots)}$), or ($\mathsf{u,del,\DOID}$) or ($\mathsf{u,w}$) for Update (add), Update (delete) and Search queries,
respectively, where parameter $\mathsf{u}$ denotes the timestamp of issuing a query. We define a function $\mathsf{F}$ of the inputs $(\mathsf{\DOID,w})$ as a randomization function, that outputs a random element for each pair of ($\mathsf{\DOID,w}$). We also define a function $\mathsf{T}$ of the input $\mathsf{w}$ as a randomization function, that outputs a random element for each $\mathsf{w}$. The definitions of the leakages are as follows.


\begin{itemize}[noitemsep,topsep=0pt,leftmargin=10pt]

\item When adding several $\DOID$s and their relative keywords in a batch insertion, function $\mathsf{N_{\DOID}}$ returns the total number of $\DOID$s that have been added.

\begin{center}
{$
\mathsf{N_{\DOID}}(\mathsf{add})=\left\{\left(\text{Number of added $\DOID$s in one batch insertion}\right)
\right\}
$}
\end{center}

\item When adding several $\DOID$s and their relative keywords in a batch insertion, function $\mathsf{NW_{\DOID}}$ returns the total number of $\mathsf{w}$s that have been added for particular $\DOID$.

\begin{center}
{$\mathsf{NW_{\DOID}}(\mathsf{add})=\{(\text{Number of added $\mathsf{w}$s for one particular $\DOID$}$ in a batch insertion})$\}$
\end{center}

\item Given an identifier $\DOID$, function $\operatorname{AddHist}(\DOID)$
returns the history timestamp of Add operation about
$\DOID$ that has been added in a batch insertion with some other $\DOID$s.
\begin{center}
{$
\operatorname{AddHist}(\DOID)=\left\{\left(\mathsf{u^{a d d}}\right) \mid \exists \text{set of }\DOID\text{s},\left(\mathsf{u^{a d d}, a d d},(\text{set of }\DOID\text{s including }\DOID)\right) \in \mathsf{Q}\right.\}$}
\end{center}


\item Given an identifier $\DOID$, function $\operatorname{DelHist}(\DOID)$
returns the history timestamps of all paired Add and Delete operations about
$\DOID$.
\begin{center}
{$
\operatorname{DelHist}(\DOID)=\{\left(\mathsf{u^{a d d}, u^{ {del }}}\right) \mid \exists \DOID,\left(\mathsf{u^{a d d}, a d d},(\text{set of }\DOID\text{s including }\DOID)\right) \in \mathsf{Q}$ $\text { and } \left(\mathsf{u^{\text {del }}, { del }},\DOID\right) \in \mathsf{Q}\}$}
\end{center}

\item Given an identifier $\DOID$, function $\mathsf{Delindex}(\DOID)$
returns the correlation of stored deltas in $\mathsf{FSet}$ with the search indices in $\mathsf{ISet}$, that is revealed after deletion of $\DOID$.

\begin{center}
{$
\mathsf{Delindex}(\DOID)=\left\{\Delta 2 \ell\text{: matching delta with search index after deleting $\DOID$}
\right\}$}
\end{center}

\item Given an identifier $\DOID$, function $\operatorname{Delw}(\DOID)$
returns a set of all $\mathsf{T}_i\mathsf{(w)}$ for all $\mathsf{w}_i$ that have been deleted in $\mathsf{u^{del}}$ and have been searched in time $\mathsf{u_i<u^{del}}$. Otherwise, returns nothing. Note: this information can be derived from query-info and from the defined set of $S$.

\begin{center}
{$
\operatorname{Delw}(\DOID)=\{\mathsf{\{T}_i({\mathsf{w_{\DOID}}})\}_i \mid  
\left(\mathsf{u^{ {del }},  { del }},\DOID\right) \in \mathsf{Q} 
\text{ and }\mathsf{w}_i\text{ has been searched before }\mathsf{u^{ {del} }}
\}$}
\end{center}


\item Given a keyword $\mathsf{w}$, function $\operatorname{sp}(\mathsf{w})$ returns all timestamps of the Search queries about
keyword $\mathsf{w}$ and $\operatorname{rp}(\mathsf{w})$ returns the timestamps and the randomized output related to the IDs returned in the search of $\mathsf{w}$.
\begin{center}
{$
\operatorname{sp}(\mathsf{w})=\{\mathsf{u \mid(u,w) \in Q}\}$}
\end{center}

\begin{center}
{$
\operatorname{rp}(\mathsf{w})=\{(\mathsf{u,F(\DOID,w)) \mid(u,w) \in Q}\}$}
\end{center}

\item Given a keyword $\mathsf{w}$, function $\operatorname{TimeDB}(\mathsf{w})$ returns all F outputs related to the undeleted identifiers ($\DOID$s) that have keyword
$\mathsf{w}$ and the history timestamps for adding these $\DOID$s.
\begin{center}
{$
\operatorname{TimeDB}(\mathsf{w})=\{(\mathsf{u, F(\DOID,w)) \mid(u, a d d,(\DOID)) \in Q}$ $\text { and } \forall \mathsf{u^{\prime},\left(u^{\prime},  { del },(\DOID)\right) \notin Q}\} $}
\end{center}

\item Given a keyword $\mathsf{w}$, skipped tokens returns all the search tokens for $\mathsf{w}$ that were deleted before the time of search for $\mathsf{w}$.

\begin{center}
{$
\operatorname{skipped~tokens}(\mathsf{w})=\{(\mathsf{u^{ {srch} }},\mathsf{F({ID,w}})) \mid  
\left(\mathsf{u^{ {del }},  { del }},\DOID\right) \in \mathsf{Q} 
\text{ and }\DOID \text{ that has }\mathsf{w}\text{ has been deleted before }\mathsf{u^{ {srch} }}
\}$}
\end{center}

\item Given an identifier $\DOID$, function $\operatorname{BFF}(\DOID)$
returns the set of entries (i.e., indices, deltas, encrypted IDs) that have been added in one batch insertion and have not been deleted yet.

\begin{center}
{$
\operatorname{BFF}(\DOID)=\{(\text{set of entries related to }\DOID s, \mathsf{w}\text{ in the database}) \mid$ $\operatorname{AddHist}(\DOID s)=\operatorname{AddHist}(\DOID)
\}$}
\end{center}

\end{itemize}

 \newtheorem{simulator}{Simulator}

In this article, ID-based DSSE (IDDSSE) is considered as a dynamic SSE that offers updates based on the IDs. It means IDs with relevant keywords are either added or deleted in the update phase. We define the below definitions of IDFP and IDBP.\par

\begin{definition}
An IDDSSE scheme is ID-forward-private if Update (add) queries do not leak which keywords are involved in the IDs that are being updated. Just the number of IDs and the total number of keywords in a batch update being added to the server are revealed.

More formally,
IDFP: A $\mathcal{L}$-non-adaptively-secure IDSSE scheme is ID-forward-private iff the update leakage function $\mathcal{L}^{\text {Updt-add}}$ can be written as:\newline
$\mathcal{L}^{Updt-\mathsf{add}}(\mathsf{add},\{\DOID_i,\mathbf{W}_{{\DOID_i}}\}_i)=$\newline$\mathcal{L}^{\prime}(\{\mathsf{add},\mathsf{NW_{\DOID_i}}(\mathsf{add}),\operatorname{AddHist}(\DOID_i)\})$ where $\mathcal{L}^{\prime}$ is stateless.


\end{definition}

\begin{definition}
An IDDSSE scheme is ID-backward-private if it does not reveal the IDs that have already been deleted but it leaks if the search on being deleted w happened before deletion, the number of IDs currently matching w, when they were inserted, and which deletion update is related to which batch insertion update.

More formally, IDBP: A $\mathcal{L}$-non-adaptively-secure IDSSE scheme is ID-backward-private iff the search and update leakage functions $\mathcal{L}^{\text {Srch }}, \mathcal{L}^{\text {Updt-del}}$ can be written as:\\
$\mathcal{L}^{Updt-\mathsf{del}}(\mathsf{del},\DOID)=\mathcal{L}^{\prime}(\{\mathsf{del},\operatorname{Delw}(\DOID),\operatorname{DelHist}(\DOID)\})$\newline
$\mathcal{L}^{Srch}(\mathsf{w})=\mathcal{L}^{\prime \prime}(\{\operatorname{sp}(\mathsf{w}),\operatorname{rp}(\mathsf{w}),\operatorname{TimeDB}(\mathsf{w})\})$

where $\mathcal{L}^{\prime}$ and $\mathcal{L}^{\prime \prime}$ are stateless.
\end{definition}

\begin{theorem}\label{theorem}
Let $\pi$ be a one-way trapdoor permutation, F a secure PRF, and (Enc,Dec) a secure symmetric encryption scheme. Assuming that the D-ACE assumption holds in $\mathbb{G}$, by defining the leakage function $\mathcal{L}$ as below, ACE is $\mathcal{L}$-non-adaptively-secure and satisfies IDFP, IDBP.\newline
$\mathcal{L}^{S t p}(\lambda)=\{\lambda\}$\newline
$\mathcal{L}^{Updt}(\mathsf{add},\{\DOID_1,\mathbf{W}_{{\DOID_1}}\},\{\DOID_2,\mathbf{W}_{{\DOID_2}}\},\ldots)=$\newline
$\{\mathsf{add},\mathsf{N_{\DOID}}(\mathsf{add}),\mathsf{NW_{\DOID}}(\mathsf{add}),\operatorname{AddHist}(\DOID)\}$\newline   
$\mathcal{L}^{Updt}(\mathsf{del},\DOID)=\{\mathsf{del},\operatorname{Delw}(\DOID),\operatorname{DelHist}(\DOID),\operatorname{BFF}(\DOID)\}$\newline
$\mathcal{L}^{Srch}(\mathsf{w})=\{\operatorname{sp}(\mathsf{w}),\operatorname{rp}(\mathsf{w}),\operatorname{TimeDB}(\mathsf{w}),\text{skipped tokens}\}$


\end{theorem}

\begin{proof}
    The proof is discussed in the Appendix \ref{appendix:proof}.
\end{proof}

Discussion: It is important to note that in our system model, $\Trustee$ and $\Vetter$ are two different entities performing their own mentioned responsibilities discussed in section \ref{section:system model}. So, the $\User$ interacts with the $\Vetter$, which does not have write permission (like the $\Trustee$ has), and in the worst case, the user might get more information but does not interact with an entity to write something or tamper with the database. Moreover, following the principle of separation of privileges, all the privileges are not granted to one entity and $\Trustee$ and $\Vetter$ are separated. Therefore, if one is compromised, the other one will not be affected. 
Additionally, it is worth mentioning that information leakages in secure searchable encryption (SSE) schemes can be mitigated through the use of oblivious RAM (ORAM) techniques \cite{roche2016practical,garg2016tworam}. However, ORAM introduces significant computational overhead and bandwidth costs for each keyword search, rendering it impractical for achieving efficient SSE. As a result, a practical SSE scheme often needs to strike a balance between information leakage and efficiency, accepting a certain degree of leakage to achieve acceptable performance.

\section{Analytical Performance Comparison}\label{sec:analytical}
This section presents the analytical performance comparison of our ACE with existing related works from different perspectives. The overall comparison is depicted in Table \ref{table:comm/comp complexity}.

\begin{table*}[ht!]
\caption{Computational and communication costs}
\label{table:comm/comp complexity}
\centering
\begin{tabular}{cc>{\centering\arraybackslash}m{3.5cm}>{\centering\arraybackslash}m{3.5cm}>{\centering\arraybackslash}m{3.5cm}}
\toprule
\multicolumn{2}{c}{\textbf{Reference} }      & \begin{tabular}[c]{@{}c@{}}ACE  \end{tabular}                           & \begin{tabular}[c]{@{}c@{}}  \cite{chen2021bestie} \end{tabular} & \begin{tabular}[c]{@{}c@{}} \cite{xu2017dynamic}\end{tabular} \\  
\cmidrule(lr){1-2}\cmidrule(lr){3-3} \cmidrule(lr){4-4} \cmidrule(lr){5-5}
\cmidrule{1-2}\cmidrule(lr){3-3} \cmidrule(lr){4-4} \cmidrule(lr){5-5}
\multicolumn{1}{c|}{}                                 & Addition                 & $x(2T_F+2T_e+T_S+T_h) + 2T_F+T_E$                                            & $x(2T_F+T_h+T_X)+T_E$                                                                                                  & $x(3T_F+3T_h+3T_X)+T_F$                                                                   \\ \cline{2-5}
\multicolumn{1}{c|}{}                                 & Deletion      & $2T_F+x(T_e+T_h)$                                                  & $x(2T_F+T_h+T_X)+T_E$                                                                                           & $2T_F+x(2T_h+3T_X+T_R)$                                                                 \\ \cline{2-5} 
\multicolumn{1}{c|}{\multirow{-4}{*}{\textbf{Comp.}}} & Search       & $T_F +\alpha(T_e+T_h+T_S)$                                            & \begin{tabular}[c]{@{}c@{}}$2T_F+T_h+$\\$N_U*(T_h+T_D)$\end{tabular}                                                                                                      & \begin{tabular}[c]{@{}c@{}}$2T_F+\alpha(T_h+T_X)+$\\ $N_D* (T_h+3T_X)$\end{tabular}                              \\ \hline
\multicolumn{1}{c|}{\textbf{Stor.}}                   & Storage Size          & $r(\ell_F+x(\ell_E+\ell_D+\ell_h))$                                                        & \begin{tabular}[c]{@{}c@{}}$rx(2\ell_h+\ell_E)+$\\$N_U*(2\ell_h+\ell_E)$\end{tabular}                                                                                                  & $3rx(\ell_h+\ell_F)$                                                                                 \\ \hline
\multicolumn{1}{c|}{\textbf{}}                   & Addition             & \begin{tabular}[c]{@{}c@{}}$\ell_F+x(\ell_E+\ell_D+\ell_h)$\end{tabular} & $x(2\ell_h+\ell_E)$                                                             & $3x(\ell_F+\ell_h)$                                                         \\ 

\cline{2-5} 
\multicolumn{1}{c|}{}                                 & Deletion       & $2\ell_F$                                                  & $x(2\ell_h+\ell_E)$                                                                                          & $2\ell_F$                                                                 \\ \cline{2-5} 
\multicolumn{1}{c|}{\multirow{-4}{*}{\textbf{Comm.}}} & Search      & $\ell_F+\ell_D$                                            & $\ell_F+\ell_h$                                                                                                     & $2\ell_F$                              \\ 

\bottomrule

\multicolumn{5}{p{15cm}}{\footnotesize{
\textbf{Notations}: $T_e$: Time needed to compute an exponentiation; $T_F$: Time needed to compute a PRF; $T_h$: Time needed to compute a hash; $T_E$: Time needed to encrypt a block with a symmetric cryptosystem; $T_S$: Time needed to compute trapdoor permutation; $T_X$: Time needed to compute XOR; $T_R$: Time needed to overwrite an entry; $N_U$: Number of updates; $N_D$: Number of deletions; $\alpha$: Number of records satisfying searched keyword; $x$: Number of keywords of an ID; $r$: Number of records in DB; $\ell_D$: Size of an element from Diffie-Hellman (DH) group; $\ell_F$: Size of the output of a PRF; $\ell_E$: Size of the block of SE; $\ell_h$: Size of the output of hash function H.}}
\end{tabular}
\end{table*}

\begin{itemize}[noitemsep,topsep=0pt,leftmargin=10pt]
    \item Update-Addition: When adding one ID (with its all relevant keywords) to the database, the computation that is needed and the communication complexity are in the order of the number of keywords the ID has for ACE, \cite{chen2021bestie} and \cite{xu2017dynamic}. If we add n IDs with their keywords, the computation and communication complexity also increases by the number of IDs in ACE, \cite{chen2021bestie} and \cite{xu2017dynamic}.

    \item Update-Deletion: To delete an ID with the relevant keywords, the computation is in the order of the number of keywords for ACE, \cite{chen2021bestie} and \cite{xu2017dynamic}. However, the communication complexity is in the order of the number of keywords for \cite{chen2021bestie} and is a small token for ACE and \cite{xu2017dynamic}.

    \item Search: Search computation complexity is in the order of the number of matched IDs in ACE, and it depends on the number of updates that has happened before search on w in \cite{chen2021bestie}. In \cite{xu2017dynamic}, the search complexity is in the order of the number of matched IDs for the keyword that is searched, and if a deletion happened before search, it needs to perform some computations to remove data in the search phase. However, ACE completes the update (addition and deletion) in their own phase and do not postpone any parts of update to the search phase.

    \item Storage: Storage size is in the order of the number of records multiplied by the number of keywords. It means it is in the order of the number of all pairs of (ID,w) in the dataset for all three schemes in Table \ref{table:comm/comp complexity}.

\end{itemize}

This analytical comparison highlights the efficiency of ACE in terms of its search and update mechanisms. While both ACE and the other scheme in \cite{xu2017dynamic} offer deletion based on ID, ACE stands out by providing instant real deletion without any negative implications. Additionally, ACE ensures low communication costs for both search and deletion operations.

\section{Experimental evaluations}\label{sec:experiments}
\subsection{Implementation}
We implemented ACE and evaluated it using different datasets. The programming environment, configuration, used cryptographic primitives, and the dataset information are as follows.\par
The hardware and software configuration used for the evaluation are as follows:
Hardware Platform: CPU: Intel i7-11850H; 
Memory: 64GB; 
Operating System: Fedora 35 x64; 
Compiler: Java 16; 
Cryptographic Library: Bouncy Castle; 
Database: Redis.

Programming Environment: We used an in-memory key-value database Redis \cite{Redis} to store FSet and ISet to improve the query and update performance. Our code is published at Proton Drive\footnote{ACE implementation: online at \href{https://drive.proton.me/urls/KZJMSC639G\#HqHLc9xGCUp1}{https://drive.proton.me/urls/ACE}}.\par
Cryptographic primitive: For all cryptographic primitives, we’ve utilised the libraries provided the Bouncy Castle \cite{bouncy}. For Pseudorandom Function PRF, we chose an AES-128 based CMAC algorithm to provide encryption for this hash function, and for $PRF_p$, a SHA-512 based HMAC was applied. For the Trapdoor Permutation $\pi$, we applied RSA-2048 cryptosystem to realise the asymmetric encryption with the characteristics of a trapdoor permutation.\par

The dataset we used to test our protocol, ACE, is a genomic dataset that part of it is a real-life dataset, which comes from The Harvard Personal Genome Project (PGP) \cite{PGP}. This is the SNP information of the patients alongside their phenotype, gender and ethnicity. By using this real-life dataset, we created synthetic datasets to evaluate ACE on datasets with different numbers of records and keywords (total number of (ID,w) pair from $5*10^4$ to $4*10^6$) to analyze its performance.


\subsection{Evaluation results}
The update, search time and communication costs, and storage analysis are discussed in this section.

\begin{itemize}[noitemsep,topsep=0pt,leftmargin=10pt]
\item{Update-Addition:}
As presented in Figure \ref{fig:Add} (a), since addition in ACE happens as a batch insertion, we evaluated the time for adding two IDs to the database when the number of keywords increases. The number of keywords of the IDs that are being added affects the time cost of the addition.

The communication cost is the size of the encrypted data that is being added to the database. So, it increases by the number of pairs of (ID,w) that are being added to the database. Figure \ref{fig:Add} (b) presents the ciphertext size when 2 IDs with different keywords are added to the FSet and ISet ($\#$pairs (w,ID)=$2*$($\#$keywords) in this evaluation).

\begin{figure*}[!t]
\centering
\subfigure[Update-Addition Time]{
\includegraphics[scale=0.48]{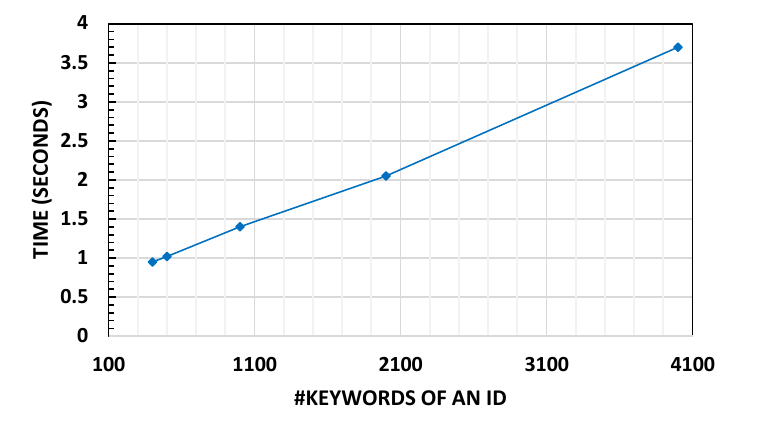}
}
\subfigure[Update-Addition ciphertext size (FSet and ISet)]{
\includegraphics[scale=0.48]{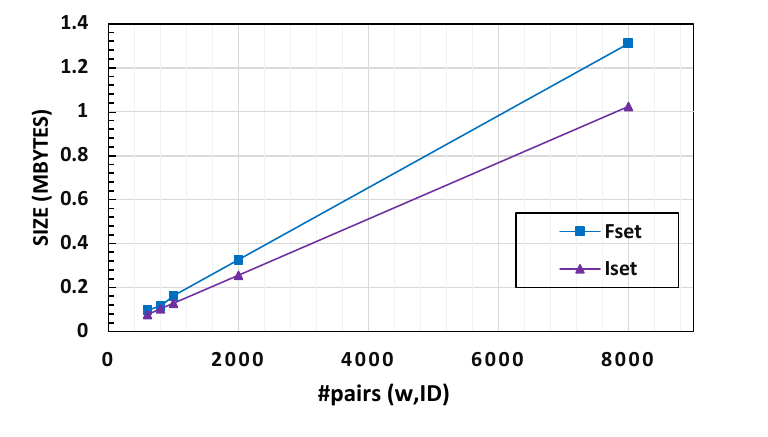}
}
\caption{Update-Addition of 2 $\DOID$s with different number of keywords}
\label{fig:Add}
\end{figure*}

\begin{figure*}[!t]
\centering
\subfigure[Update-Deletion Time (Server and Vetter)]{
\includegraphics[scale=0.5]{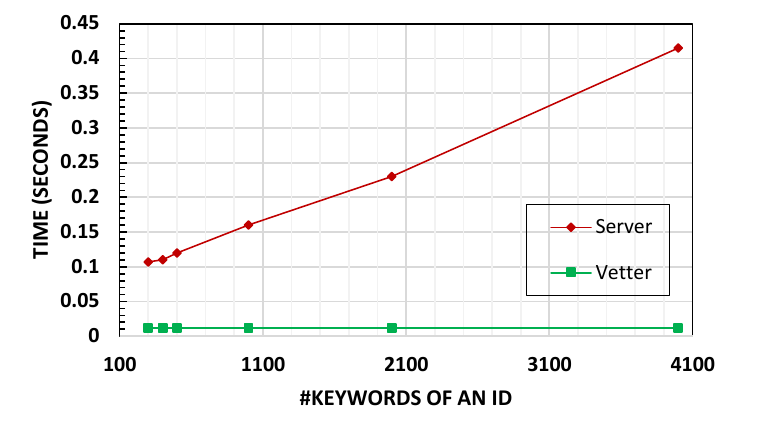}
}
\subfigure[Update-Deletion token size in ACE and Bestie]{
\includegraphics[scale=0.5]{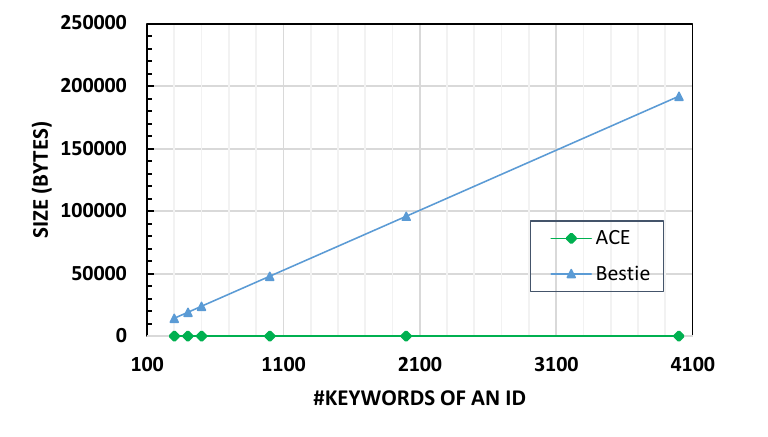}
}
\caption{Update-Deletion of 1 $\DOID$ with different number of keywords}
\label{fig:Del}
\end{figure*}

\item{Update-Deletion}: 
When a consent is revoked, or whenever the data of a data owner needs to be removed from the database, the Update-del algorithm removes the relevant data of a data owner. For this type of deletion, when the number of keywords of a data owner increases, the deletion time increases. However, since the vetter generates one token for deleting all the data of a data owner, the vetter's computation complexity is constant (see Figure \ref{fig:Del} (a)). The deletion of data of an ID happens in a non-interactive fashion (one token sent from the vetter to the server).

Since ACE provides deletion based on an ID, the deletion token is constant in size when the data of a data owner needs to be deleted. The number of keywords does not have any effect on the size of the token; hence, the required bandwidth does not increase for IDs with different number of keywords. However, in the schemes that support deletion of pair of (ID,w), the required bandwidth for deleting the data of a data owner increases by the number of keywords the data owner has. This is because for each keyword, a new token needs to be generated and sent to the server. This behaviour is shown in Figure \ref{fig:Del} (b) and the token size for Bestie protocol in \cite{chen2021bestie} is calculated from the sizes discussed in their paper.
The provided graphs' trends are consistent with analytical discussion in section \ref{sec:analytical}.

\item{Search:}
The search time in \cite{chen2021bestie} and \cite{xu2017dynamic} depends on the number of updates that were done before the search, as these two schemes do not complete the deletion when the deletion query is performed. They complete removing the data in the search phase. To provide instant deletion (data is deleted when it is requested), we do not postpone deletion or part of it to the search phase. The search time in ACE increases with the number of matched IDs for the keyword that is searched (see Figure \ref{fig:Search}). There is also an initialization time cost of around $200$ms due to Java processing that is included in the presented search results.
\begin{figure}[b]
\centering
\includegraphics[scale=0.48]{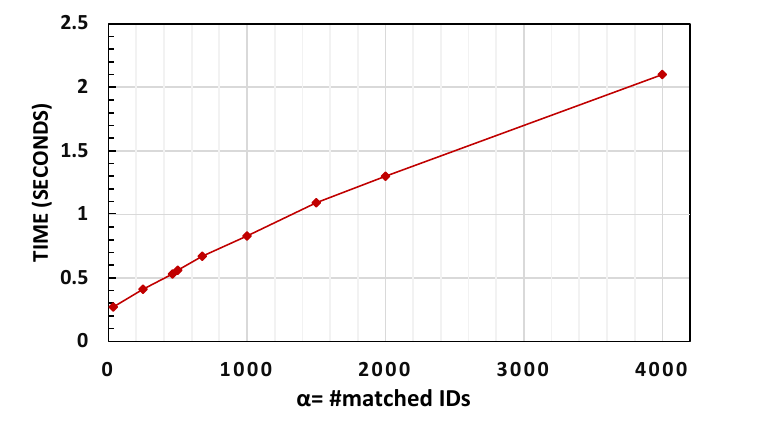}
\caption{Total search time with different number of matched $\DOID$s}
\label{fig:Search}
\end{figure}

\newtheorem{rmk}{Remark}
\begin{rmk}
    The main advantages of ACE are providing the features of instant deletion when the consent is revoked with low deletion communication complexity (one deletion token/non-interactive) and privacy of the ID (these are discussed in Table \ref{table:first comparison}). In terms of search time, we show that we achieve all these advantages with a reasonable performance (Figure \ref{fig:Search}). Therefore, we extract search time of other schemes and show that although ACE does not have the best search time, it still has a reasonable performance in comparison to earlier schemes that do not support the mentioned features of ACE. As it is shown in Table \ref{table:Search}, when the number of matched IDs ($\alpha$) is $200$, the search time of ACE is $0.38$s, and when $\alpha$ is $2,000$, the ACE search time is $1.3$s that is $10$ times and $700$ times speedup in comparison to Janus++ and Janus evaluated in \cite{sun2018practical}. It is also worth mentioning that in comparison with the schemes coded using C++ such as \cite{chen2021bestie}, we have slow down in results due to the compiler Java.
\end{rmk}

\begin{table}[t]
\caption{Search time in different schemes with different number of matched $\DOID$s}
\label{table:Search}
\centering
\begin{tabular}{>{\centering\arraybackslash}m{1.5cm}>{\centering\arraybackslash}m{2.5cm}>{\centering\arraybackslash}m{2.5cm}}
\toprule
\multirow{1}{*}{$\alpha$}$\ssymbol{5}$   & Scheme       & Search time$\ssymbol{4}$        \\\midrule

\multirow{3}{*}{$200$}   & \cite{xu2017dynamic}$\ssymbol{1}$       & $400$ ms           \\ \cmidrule(lr){2-3}
& \cite{chen2021bestie} & $<200$ ms\\ \cmidrule(lr){2-3}
& ACE & $380$ ms
 \\ \midrule
\multirow{3}{*}{$2,000$}& \cite{bost2017forward}$\ssymbol{2}$ & $700$ s\\\cmidrule(lr){2-3}
& \cite{sun2018practical}$\ssymbol{3}$ & $10$ s\\\cmidrule(lr){2-3}
& ACE & $1.3$ s      \\[1pt] \bottomrule

\multicolumn{3}{p{7cm}}{\footnotesize{$\ssymbol{5}$The comparison with different schemes is presented with different number of matched IDs since the results are extracted from the cited papers and they evaluated their schemes with different $\alpha$s; $\ssymbol{4}$: These are approximate times that are extracted from the schemes' provided graphs in their papers; $\ssymbol{1}$: Data is extracted from their paper with $|\mathsf{DB}|=10^7$; $\ssymbol{2}$: Data is extracted from \cite{sun2018practical} with number of deletions$=100$ for the Janus protocol in this paper; $\ssymbol{3}$: number of deletions$=100$.}}

\end{tabular}
\end{table}

\begin{table}[t]
\caption{Storage size (original, encrypted FSet, ISet on Server, and W on Vetter) for $1,000$ number of $\DOID$s with different number of keywords}
\label{table:Storage cost}
\centering
\begin{tabular}{ccccc}
\toprule
\multirow{1}{*}{\#Keywords(x)}   
 & Original  & FSet  & ISet  & W  \\ \midrule
$500$       & $2.8$ MB         & $82$ MB        & $68$ MB       & $0.76$ MB               \\[1pt] \hline
$1,000$       & $5.7$ MB         & $164$ MB        & $137$ MB       & $1.5$ MB               \\[1pt] \hline
$4,000$      & $26.2$ MB       & $656$ MB        & $546$ MB       & $6$ MB                 \\[1pt] \bottomrule
\end{tabular}
\end{table}

\item{Storage:} The storage cost on the server side (FSet, ISet), and on the vetter side (W) are presented in Table \ref{table:Storage cost}. The results are for different datasets with $1,000$ number of $\DOID$s and different number of keywords. The storage size on the server side increases when the number of $\DOID$s or the number of keywords of an $\DOID$ increases, but the size of W depends only on the number of distinct keywords in the dataset.

\end{itemize}

\section{Conclusion}\label{conclusion}

In this paper, we introduce our novel scheme called ACE, which addresses the challenges of consent revocation and non-interactive instant deletion based on the data owner's identifier (ID). ACE achieves this by implementing physical deletion of a data owner's information at the moment their consent is revoked. By promptly removing the data instead of retaining it for later deletion, ACE ensures compliance with privacy regulations and mitigates potential privacy concerns.
Moreover, we define a hard problem, D-ACE, and prove its hardness by a security reduction from DDH to D-ACE. We present two new definitions of ID-based forward privacy (IDFP) and ID-based backward privacy (IDBP). Hence, we use these tools to facilitate our formal security proof of ACE. Finally, we evaluate ACE using real-life and synthetic genomic datasets and show its performance and applicability while providing the advantage of IDFP/IDBP in our scheme, with an instant deletion based on ID.

\remove{

\begin{acks}
To Robert, for the bagels and explaining CMYK and color spaces.
\end{acks}

}


\bibliographystyle{unsrt}
\bibliography{main-base}

\begin{thebibliography}{10}

\bibitem{kaye2015dynamic}
Jane Kaye, Edgar~A Whitley, David Lund, Michael Morrison, Harriet Teare, and
  Karen Melham.
\newblock Dynamic consent: a patient interface for twenty-first century
  research networks.
\newblock {\em European journal of human genetics}, 23(2):141--146, 2015.

\bibitem{gdpr}
Protection Regulation.
\newblock Regulation (eu) 2016/679 of the european parliament and of the
  council.
\newblock {\em Regulation (eu)}, 679:2016, 2016.

\bibitem{budin2017dynamic}
Isabelle Budin-Lj{\o}sne, Harriet~JA Teare, Jane Kaye, Stephan Beck,
  Heidi~Beate Bentzen, Luciana Caenazzo, Clive Collett, Flavio D’Abramo,
  Heike Felzmann, Teresa Finlay, et~al.
\newblock Dynamic consent: a potential solution to some of the challenges of
  modern biomedical research.
\newblock {\em BMC medical ethics}, 18(1):1--10, 2017.

\bibitem{prictor2020dynamic}
Megan Prictor, Megan~A Lewis, Ainsley~J Newson, Matilda Haas, Sachiko Baba,
  Hannah Kim, Minori Kokado, Jusaku Minari, Fruzsina Molnar-Gabor, Beverley
  Yamamoto, et~al.
\newblock Dynamic consent: an evaluation and reporting framework.
\newblock {\em Journal of Empirical Research on Human Research Ethics},
  15(3):175--186, 2020.

\bibitem{CIC}
Sara Jafarbeiki, Raj Gaire, Amin Sakzad, Shabnam~Kasra Kermanshahi, and Ron
  Steinfeld.
\newblock Collaborative analysis of genomic data: vision and challenges.
\newblock In {\em 2021 IEEE 7th International Conference on Collaboration and
  Internet Computing (CIC)}, pages 77--86, 2021.

\bibitem{erlich2014redefining}
Yaniv Erlich, James~B Williams, David Glazer, Kenneth Yocum, Nita Farahany,
  Maynard Olson, Arvind Narayanan, Lincoln~D Stein, Jan~A Witkowski, and
  Robert~C Kain.
\newblock Redefining genomic privacy: trust and empowerment.
\newblock {\em PLoS biology}, 12(11):e1001983, 2014.

\bibitem{erlich2014routes}
Yaniv Erlich and Arvind Narayanan.
\newblock Routes for breaching and protecting genetic privacy.
\newblock {\em Nature Reviews Genetics}, 15(6):409--421, 2014.

\bibitem{sun2018practical}
Shi-Feng Sun, Xingliang Yuan, Joseph~K Liu, Ron Steinfeld, Amin Sakzad, Viet
  Vo, and Surya Nepal.
\newblock Practical backward-secure searchable encryption from symmetric
  puncturable encryption.
\newblock In {\em Proceedings of the 2018 ACM SIGSAC Conference on Computer and
  Communications Security}, pages 763--780, 2018.

\bibitem{sun2021practical}
Shi-Feng Sun, Ron Steinfeld, Shangqi Lai, Xingliang Yuan, Amin Sakzad, Joseph~K
  Liu, Surya Nepal, and Dawu Gu.
\newblock Practical non-interactive searchable encryption with forward and
  backward privacy.
\newblock In {\em NDSS}, 2021.

\bibitem{stefanov2013practical}
Emil Stefanov, Charalampos Papamanthou, and Elaine Shi.
\newblock Practical dynamic searchable encryption with small leakage.
\newblock {\em Cryptology ePrint Archive}, 2013.

\bibitem{xu2017dynamic}
Peng Xu, Shuai Liang, Wei Wang, Willy Susilo, Qianhong Wu, and Hai Jin.
\newblock Dynamic searchable symmetric encryption with physical deletion and
  small leakage.
\newblock In {\em Australasian Conference on Information Security and Privacy},
  pages 207--226. Springer, 2017.

\bibitem{chen2021bestie}
Tianyang Chen, Peng Xu, Wei Wang, Yubo Zheng, Willy Susilo, and Hai Jin.
\newblock Bestie: Very practical searchable encryption with forward and
  backward security.
\newblock In {\em European Symposium on Research in Computer Security}, pages
  3--23. Springer, 2021.

\bibitem{bost2017forward}
Rapha{\"e}l Bost, Brice Minaud, and Olga Ohrimenko.
\newblock Forward and backward private searchable encryption from constrained
  cryptographic primitives.
\newblock In {\em Proceedings of the 2017 ACM SIGSAC Conference on Computer and
  Communications Security}, pages 1465--1482, 2017.

\bibitem{cash2015leakage}
David Cash, Paul Grubbs, Jason Perry, and Thomas Ristenpart.
\newblock Leakage-abuse attacks against searchable encryption.
\newblock In {\em Proceedings of the 22nd ACM SIGSAC conference on computer and
  communications security}, pages 668--679, 2015.

\bibitem{blackstone2019revisiting}
Laura Blackstone, Seny Kamara, and Tarik Moataz.
\newblock Revisiting leakage abuse attacks.
\newblock {\em Cryptology ePrint Archive}, 2019.

\bibitem{zhang2016all}
Yupeng Zhang, Jonathan Katz, and Charalampos Papamanthou.
\newblock All your queries are belong to us: the power of
  $\{$File-Injection$\}$ attacks on searchable encryption.
\newblock In {\em 25th USENIX Security Symposium (USENIX Security 16)}, pages
  707--720, 2016.

\bibitem{bost2016ovarphiovarsigma}
Raphael Bost.
\newblock $\sigma$ o$\varphi$o$\varsigma$: Forward secure searchable
  encryption.
\newblock In {\em Proceedings of the 2016 ACM SIGSAC Conference on Computer and
  Communications Security}, pages 1143--1154, 2016.

\bibitem{gdprerasure}
Protection Regulation.
\newblock Regulation (eu) 2016/679 of the european parliament and of the
  council-art. 17.
\newblock {\em Regulation (eu)}.

\bibitem{ghareh2018new}
Javad Ghareh~Chamani, Dimitrios Papadopoulos, Charalampos Papamanthou, and
  Rasool Jalili.
\newblock New constructions for forward and backward private symmetric
  searchable encryption.
\newblock In {\em Proceedings of the 2018 ACM SIGSAC Conference on Computer and
  Communications Security}, pages 1038--1055, 2018.

\bibitem{zuo2019dynamic}
Cong Zuo, Shi-Feng Sun, Joseph~K Liu, Jun Shao, and Josef Pieprzyk.
\newblock Dynamic searchable symmetric encryption with forward and stronger
  backward privacy.
\newblock In {\em European symposium on research in computer security}, pages
  283--303. Springer, 2019.

\bibitem{zuo2021searchable}
Cong Zuo, Shangqi Lai, Xingliang Yuan, Joseph~K Liu, Jun Shao, and Huaxiong
  Wang.
\newblock Searchable encryption for conjunctive queries with extended forward
  and backward privacy.
\newblock {\em Cryptology ePrint Archive}, 2021.

\bibitem{kasra2022range}
Shabnam Kasra~Kermanshahi, Rafael Dowsley, Ron Steinfeld, Amin Sakzad, Joseph
  Liu, Surya Nepal, Xun Yi, and Shangqi Lai.
\newblock Range search on encrypted spatial data with dynamic updates.
\newblock {\em Journal of Computer Security}, (Preprint):1--21, 2022.

\bibitem{song2000practical}
Dawn~Xiaoding Song, David Wagner, and Adrian Perrig.
\newblock Practical techniques for searches on encrypted data.
\newblock In {\em Proceeding 2000 IEEE symposium on security and privacy. S\&P
  2000}, pages 44--55. IEEE, 2000.

\bibitem{goh2003secure}
Eu-Jin Goh.
\newblock Secure indexes.
\newblock {\em Cryptology ePrint Archive}, 2003.

\bibitem{curtmola2006searchable}
Reza Curtmola, Juan Garay, Seny Kamara, and Rafail Ostrovsky.
\newblock Searchable symmetric encryption: improved definitions and efficient
  constructions.
\newblock In {\em Proceedings of the 13th ACM conference on Computer and
  communications security}, pages 79--88, 2006.

\bibitem{rangequery}
W.~Sun, N.~Zhang, W.~Lou, and Y.~Th. Hou.
\newblock When gene meets cloud: Enabling scalable and efficient range query on
  encrypted genomic data.
\newblock In {\em IEEE INFOCOM 2017-IEEE Conference on Computer
  Communications}, pages 1--9. IEEE, 2017.

\bibitem{faber2015rich}
Sky Faber, Stanislaw Jarecki, Hugo Krawczyk, Quan Nguyen, Marcel Rosu, and
  Michael Steiner.
\newblock Rich queries on encrypted data: Beyond exact matches.
\newblock In {\em European symposium on research in computer security}, pages
  123--145. Springer, 2015.

\bibitem{cash2013highly}
D.~Cash, S.~Jarecki, C.~Jutla, H.~Krawczyk, M-C. Ro{\c{s}}u, and M.~Steiner.
\newblock Highly-scalable searchable symmetric encryption with support for
  boolean queries.
\newblock In {\em {\em Annual cryptology conference}}, pages 353--373. 2013.

\bibitem{kermanshahi2019multi}
Shabnam~Kasra Kermanshahi, Joseph~K Liu, Ron Steinfeld, Surya Nepal, Shangqi
  Lai, Randolph Loh, and Cong Zuo.
\newblock Multi-client cloud-based symmetric searchable encryption.
\newblock {\em IEEE Transactions on Dependable and Secure Computing},
  18(5):2419--2437, 2019.

\bibitem{jafarbeiki2021privgendb}
Sara Jafarbeiki, Amin Sakzad, Shabnam~Kasra Kermanshahi, Raj Gaire, Ron
  Steinfeld, Shangqi Lai, and Gad Abraham.
\newblock Privgendb: Efficient and privacy-preserving query executions over
  encrypted snp-phenotype database.
\newblock {\em arXiv preprint arXiv:2104.02890}, 2021.

\bibitem{jafarbeiki2021non}
Sara Jafarbeiki, Amin Sakzad, Shabnam~Kasra Kermanshahi, Ron Steinfeld, Raj
  Gaire, and Shangqi Lai.
\newblock A non-interactive multi-user protocol for private authorised query
  processing on genomic data.
\newblock In {\em International Conference on Information Security}, pages
  70--94. Springer, 2021.

\bibitem{jafarbeiki2022pressgendb}
Sara Jafarbeiki, Amin Sakzad, Shabnam~Kasra Kermanshahi, Ron Steinfeld, and Raj
  Gaire.
\newblock Pressgendb: Privacy-preserving substring search on encrypted genomic
  database.
\newblock In {\em IEEE INFOCOM 2022-IEEE Conference on Computer Communications
  Workshops (INFOCOM WKSHPS)}, pages 1--6. IEEE, 2022.

\bibitem{kamara2012dynamic}
Seny Kamara, Charalampos Papamanthou, and Tom Roeder.
\newblock Dynamic searchable symmetric encryption.
\newblock In {\em Proceedings of the 2012 ACM conference on Computer and
  communications security}, pages 965--976, 2012.

\bibitem{cash2014dynamic}
David Cash, Joseph Jaeger, Stanislaw Jarecki, Charanjit Jutla, Hugo Krawczyk,
  Marcel-C{\u{a}}t{\u{a}}lin Ro{\c{s}}u, and Michael Steiner.
\newblock Dynamic searchable encryption in very-large databases: Data
  structures and implementation.
\newblock {\em Cryptology ePrint Archive}, 2014.

\bibitem{naveed2014dynamic}
Muhammad Naveed, Manoj Prabhakaran, and Carl~A Gunter.
\newblock Dynamic searchable encryption via blind storage.
\newblock In {\em 2014 IEEE Symposium on Security and Privacy}, pages 639--654.
  IEEE, 2014.

\bibitem{prf}
J.~Katz and Y.~Lindell.
\newblock {\em Introduction to modern cryptography book}.
\newblock In {\em CRC press}, 2020.

\bibitem{boneh2020graduate}
Dan Boneh and Victor Shoup.
\newblock A graduate course in applied cryptography.
\newblock {\em Draft 0.5}, 2020.

\bibitem{roche2016practical}
Daniel~S Roche, Adam Aviv, and Seung~Geol Choi.
\newblock A practical oblivious map data structure with secure deletion and
  history independence.
\newblock In {\em 2016 IEEE Symposium on Security and Privacy (SP)}, pages
  178--197. IEEE, 2016.

\bibitem{garg2016tworam}
Sanjam Garg, Payman Mohassel, and Charalampos Papamanthou.
\newblock Tworam: efficient oblivious ram in two rounds with applications to
  searchable encryption.
\newblock In {\em Advances in Cryptology--CRYPTO 2016: 36th Annual
  International Cryptology Conference, Santa Barbara, CA, USA, August 14-18,
  2016, Proceedings, Part III}, pages 563--592. Springer, 2016.

\bibitem{Redis}
Redis Labs.
\newblock Redis.
\newblock 2017.

\bibitem{bouncy}
Bouncycastle.
\newblock The legion of the bouncy castle.
\newblock 2022.

\bibitem{PGP}
Harvard Medical School. \url{}, 
  {https://pgp.med.harvard.edu/data}, The Personal Genome Project.

\end{thebibliography}


\newpage
\appendix

\section{Background}\label{background}
\begin{definition} (Forward and Type-III-Backward privacy). An $\mathcal{L}$-adaptively secure DSSE scheme $\Sigma$ is forward-and-Type-III-backward private iff the leakage functions of Update and Search, $\mathcal{L}^{\textup{Update}}$ and $\mathcal{L}^{\textup{Search}}$ can be written as

\begin{center}
    
$
\mathcal{L}^{\textup{Update}}(\mathsf{o p, w, i d})=\mathcal{L}^{\prime}(\mathsf{o p})$ and $\mathcal{L}^{\textup{Search}}(\mathsf{w})=\mathcal{L}^{\prime \prime}(\operatorname{sp}(\mathsf{w}), \operatorname{TimeDB}(\mathsf{w}), \operatorname{DelHist}(\mathsf{w}))
$
\end{center}
where $\mathcal{L}^{\prime}$ and $\mathcal{L}^{\prime \prime}$ are two stateless functions.
\end{definition}

There are two further types of backwards privacy, referred to as Type-I and Type-II backward privacy, in addition to Type-III backward privacy. A search query only reveals the total number of updating $\mathsf{w}$ and TimeDB($\mathsf{w}$) in order to maintain Type-I backwards privacy.
The timestamps of updating $\mathsf{w}$, however, can also be leaked by a Search query when Type-II backward privacy is used.
The Type-III backward privacy has been defined as an example and the formal definitions of Type-I and Type-II backward privacy can be found in \cite{bost2017forward}. 

\section{Proof of Theorem \ref{theorem}}\label{appendix:proof}

\begin{proof}
To prove the security of our scheme, we construct a simulator, which takes as inputs leakage functions
$\mathcal{L}^{S t p}(\lambda)$,  $\mathcal{L}^{Updt}(add,$ $\{id1,W1\},$ $\{id2,$ $W2\},$ $\ldots)$ and query-info, $\mathcal{L}^{Updt}(del,id)$, $\mathcal{L}^{Srch}(w)$ to simulate protocols Setup, Update, and Search, respectively. We will demonstrate that the simulated scheme is indistinguishable from the real scheme under the non-adaptive attacks. query-info is given to the simulator at the Update phase that gives the information of $\operatorname{Delw}(id)$ -for the IDs that are selected by adversary to be deleted- to the simulator at the beginning. Algorithm \ref{alg: Simulator} describes the simulator.

For constructing the simulator, we are going to derive several games from the real world game.\newline
\textbf{Game} G$_0\hspace{4mm}$G$_0$ is exactly the real world security game depicted in Algorithm \ref{alg: game G0} and \ref{alg: game G0-continue}.
\begin{center}
$\mathbb{P}[\operatorname{Real}_{\mathcal{A}}^{\Sigma}(\lambda)=1]$=$\mathbb{P}[$G$_0=1]$
\end{center}

\begin{algorithm}[hbt!]
\caption{\textbf{Game G$_0$}}
        \label{alg: game G0}
\underline{Setup} 
This is same as Setup in Algorithm \ref{alg: Setup}
\remove{
 \begin{algorithmic}[1]

\State Select keys $\K_S, \K_1$ for PRF $F$ and ($\mathsf{SK, PK}$) for $\pi $ and keys $\K_T,~\K_2$ for PRF $F_p$ (with range in $\mathbb{Z}^*_p$) and $k_h$ for keyed hash function H
using security parameter $\lambda$, and $\mathbb{G}$ a group of prime order $p$ and generator $\mathsf{g}$.
\State Initialise empty maps $\mathbf{W}[\mathsf{w}], \mathsf{FSet}$ and empty dictionary $\mathsf{ISet}$

\remove{
 \State Parse $\DBGePh$ as $({\DOID}_i, {\mathsf{w}}_j)$ 

\For {each ${\mathsf{w}}$} 
\State $\mathsf{tag}_{\mathsf{w}} \gets F_p(\K_T,{\mathsf{w}})$; $\K_{\mathsf{w}} \gets F(\K_S,{\mathsf{w}})$.
\State Set a counter $c\gets 1$ and select $ \mathsf{S T_{0}} \stackrel{\$}{\leftarrow} \mathcal{M}$
\For {${\DOID}_i \in \DBGePh({\mathsf{w}})$}
\If {there is no index $r_{\DOID}$ in $\mathsf{FSet}$ for ${\DOID}_i$}
\State Compute index $r_{{\DOID}_i}\gets F(\K_1,{\DOID}_i)$
and a tag $\mathsf{tag}_{{\DOID}_i} \gets F_p(\K_2,{{\DOID}_i})$
\EndIf
\State Compute $\DOID' \gets {E}(\K_{\mathsf{w}},\DOID)$ 

\State $ \mathsf{S T}_{c} \leftarrow \pi_{\mathsf{S K}}^{-1}\left(\mathsf{S T}_{c-1}\right)$

\State $\ell \leftarrow H(k_h,\mathsf{g}^{{\mathsf{S T}_{c}}\cdot {\mathsf{tag}_{\mathsf{w}}}})$ 

\State Append $\DOID'$ to ${\bf \mathsf{ISet}[\ell]}$ and $c\gets c+1$\newline
\State Compute $\Delta \gets \mathsf{g}^{{\mathsf{S T}_{c}}\cdot {\mathsf{tag}_{\mathsf{w}}}/{\mathsf{tag}_{{\DOID}_i}}}$ 


\State Append $\Delta$ into ${{ \mathsf{FSet}}[r_{{\DOID}_i}]}$
\EndFor
\State $\mathbf{W}[\mathsf{w}] \leftarrow\left(\mathsf{S T}_{c}, c\right)$
\EndFor

}

\State $\mathsf{E}\DBGePh1 =\mathsf{FSet}, \mathsf{E}\DBGePh2 = \mathsf{ISet}$.
\State \Return{$\mathsf{E}\DBGePh1,\mathsf{E}\DBGePh2$} 
, and keep $\mathbf{W}[\mathsf{w}]$
 and {$\K=(\K_S, \K_T, \K_1, \K_2,k_h)$} 
\end{algorithmic}

}


\underline{Update-add (a set of $\DOID$s with their keywords, $\{\DOID_i,\mathbf{W}_{\DOID_i}\}$)}
 \begin{algorithmic}[1]
\State Parse the set as $({\DOID}_i, {\mathsf{w}}_j)$
\For {each ${\mathsf{w}}$} 
\State $\mathsf{tag}_{\mathsf{w}} \gets F(\K_T,{\mathsf{w}})$; $\K_{\mathsf{w}} \gets F(\K_S,{\mathsf{w}})$.//specific ${\DOID}_i$

\State $\left(\mathsf{S T}_{c}, c\right) \leftarrow \mathbf{W}[\mathsf{w}]$
\If {$\left(\mathsf{S T}_{c}, c\right)=\perp$}
\State $\quad \mathsf{S T_{0}} \stackrel{\$}{\leftarrow} \mathcal{M}, c \leftarrow0$
\EndIf

\For {${\DOID}_i \in \DBGePh({\mathsf{w}})$}
\If {there is no index $r_{\DOID}$ in $\mathsf{FSet}$ for ${\DOID}_i$}
\State Compute index $r_{{\DOID}_i}\gets F(\K_1,{\DOID}_i)$
and a tag $\mathsf{tag}_{{\DOID}_i} \gets F(\K_2,{{\DOID}_i})$
\EndIf
\State Compute $\DOID' \gets {E}(\K_{\mathsf{w}},\DOID)$ 
\State $c\gets c+1$
\State $ \mathsf{S T}_{c} \leftarrow \pi_{\mathsf{S K}}^{-1}\left(\mathsf{S T}_{c-1}\right)$; $ \mathsf{S T'}_{c} \leftarrow \left(\mathsf{S T}_{c}\text{ mod }p\right)$
\State $\ell \leftarrow H(k_h,\mathsf{g}^{{\mathsf{S T'}_{c}}\cdot {\mathsf{tag}_{\mathsf{w}}}})$ 

\State Append $\DOID'$ to ${\bf \mathsf{ISet}[\ell]}$ 

\State Compute $\Delta \gets \mathsf{g}^{{\mathsf{S T'}_{c}}\cdot {\mathsf{tag}_{\mathsf{w}}}/{\mathsf{tag}_{{\DOID}_i}}}$ 

\State Append $\Delta$ into ${{ \mathsf{FSet}}[r_{{\DOID}_i}]}$
\EndFor

\State $\mathbf{W}[\mathsf{w}] \leftarrow\left(\mathsf{S T}_{c}, c\right)$
\EndFor

\end{algorithmic}

\end{algorithm}


\begin{algorithm}[hbt!]
\caption{\textbf{Game G$_0$-continue}}
        \label{alg: game G0-continue}

\underline{Update-del (all entries for a particular $\DOID_i$)}
 \begin{algorithmic}[1]

\State Compute $\mathsf{tag}_{{\DOID}_i} \gets F(\K_2,{{\DOID}_i})$, $r_{{\DOID}_i}\gets F(\K_1,{\DOID}_i)$ 
\For {all elements $\Delta_i$ in $\mathsf{FSet[r_{{\DOID}_i}]}$}
\State Compute $\ell \gets H(k_h,{\Delta_i}^{{\mathsf{tag}_{{\DOID}_i}}})$
\State Remove corresponding entry from $\mathsf{ISet}[\ell]$ and $\mathsf{\ell}$
\EndFor
\State Remove entries of $\mathsf{FSet[r_{{\DOID}_i}}]$ and $\mathsf{r_{{\DOID}_i}}$
\end{algorithmic}

\underline{Search}
\begin{algorithmic}[1]

\State Vetter computes $\mathsf{tag}_{\mathsf{w}} \gets F(\K_T,{\mathsf{w}})$, $\mathsf{tk}\gets \mathsf{g}^{\mathsf{tag}_{\mathsf{w}}}$ and gets $\left(\mathsf{S T}_{c}, c\right) \leftarrow \mathbf{W}[\mathsf{w}]$
\State $\mathsf{RSet}$ $\gets \{\}$

\If {$\left(\mathsf{S T}_{c}, c\right)=\perp$}
\State  return $\emptyset$
\EndIf
\State Send $\left(\mathsf{tk}, \mathsf{S T}_{c}, c\right)$ to the server.\newline
Server:
\For {$i=c$ to 1}
\State $ \ell \leftarrow H(k_h,\mathsf{tk}^{({\mathsf{S T}_{i}}\text{ mod }p)})$
\State $ \DOID' \leftarrow \mathsf{ISet}\left[\ell\right]$
\State $\mathsf{RSet}$ $\gets \mathsf{RSet}\cup \DOID'$
\State $\mathsf{S T}_{i-1} \leftarrow \pi_{\mathrm{PK}}\left(\mathsf{S T}_{i}\right)$
\EndFor

\State \Return{$\mathsf{RSet}$}
\end{algorithmic}

\end{algorithm}


\noindent
\textbf{Game} G$_1\hspace{4mm}$Instead of calling PRF when generating tags for w and id, G1 picks a new random tag when it is confronted to a new w and id, and stores it in a table so it can be reused next time needed. It also does the same for generating $\mathsf{K_w}$ and indices $r_\DOID$. If an adversary is able to distinguish between G0 and G1, we can then build a reduction able to distinguish between PRF F and a truly random function. More formally, there exists an efficient adversary B1 such that
\begin{center}
    $\mathbb{P}\left[G_0=1\right]-\mathbb{P}\left[G_1=1\right] \leq \operatorname{Adv}_{F, B_1}^{\mathrm{prf}}(\lambda)$
\end{center}
\textbf{Game} G$_2\hspace{4mm}$This game is similar to G1 except that we encrypt a constant 0 by using the symmetric encryption SE when encrypting the IDs. If an adversary A can distinguish G2
from G1, then we can establish an adversary B2 to break the IND-CPA security of the standard symmetric key encryption SE.
\begin{center}
    $\mathbb{P}\left[G_1=1\right]-\mathbb{P}\left[G_2=1\right] \leq \operatorname{Adv}_{SE, B_2}^{\mathrm{IND-CPA}}(\lambda)$
\end{center}



\begin{algorithm}[ht!]
\caption{\textbf{Game G$_3$, \textcolor{blue}{$\hat{\text{G}_3}$}}}
        \label{alg: game}
\underline{Update-add (a set of $\DOID$s with their keywords, $\{\DOID_i,\mathbf{W}_{\DOID_i}\}$)}
 \begin{algorithmic}[1]
\State Parse the set as $({\DOID}_i, {\mathsf{w}}_j)$
\For {each ${\mathsf{w}}$} 
\State $\mathsf{tag}_{\mathsf{w}} \stackrel{\$}{\leftarrow} \{0,1\}^\lambda$; $\K_{\mathsf{w}} \stackrel{\$}{\leftarrow} \{0,1\}^\lambda$.

\State $\left(\mathsf{S T}_{c}, c\right) \leftarrow \mathbf{W}[\mathsf{w}]$
\If {$\left(\mathsf{S T}_{c}, c\right)=\perp$}
\State $\quad \mathsf{S T_{0}} \stackrel{\$}{\leftarrow} \mathcal{M}, c \leftarrow0$
\EndIf

\For {${\DOID}_i \in \DBGePh({\mathsf{w}})$}
\If {there is no index $r_{\DOID}$ in $\mathsf{FSet}$ for ${\DOID}_i$}
\State Compute index $r_{{\DOID}_i}\stackrel{\$}{\leftarrow} \{0,1\}^\lambda$
and a tag $\mathsf{tag}_{{\DOID}_i} \stackrel{\$}{\leftarrow} \{0,1\}^\lambda$
\EndIf
\State $\DOID' \gets E(\mathsf{K_w},\{0\}^\lambda)$ 
\State $c\gets c+1$

\State $\ell_{ij}\stackrel{\$}{\leftarrow} \{0,1\}^\lambda$

\State $\mathsf{S T}_{c} \leftarrow \pi_{\mathsf{S K}}^{-1}\left(\mathsf{S T}_{c-1}\right)$; $ \mathsf{S T'}_{c} \leftarrow \left(\mathsf{S T}_{c}\text{ mod }p\right)$
\color{blue}
\If{H(k1,$\mathsf{g}^{{\mathsf{S T'}_{c}}\cdot {\mathsf{tag}_{\mathsf{w}}}})\neq \perp$}
\State bad$\gets$true; $\ell_{ij}\gets$H(k1,$\mathsf{g}^{{\mathsf{S T'}_{c}}\cdot {\mathsf{tag}_{\mathsf{w}}}})$
\EndIf
\color{black}
\If{$\DOID_i$ is in query-info to be deleted and $\mathsf{w}$ related to the $\Delta_j$ $\in \{\text{Srch}<\text{Del}\}_i$ }

\State $\Delta_j \gets {\mathsf{g}^{{({\mathsf{S T'}_{c}}\cdot {\mathsf{tag}_{\mathsf{w}}}})/\mathsf{tag_{\DOID_i}}}}$
\State program H s.t. H(k1,$\mathsf{g}^{{\mathsf{S T'}_{c}}\cdot {\mathsf{tag}_{\mathsf{w}}}}$)$\gets \ell_{ij} $
\State Keep the record of the $\mathsf{ST}$s used for $\mathsf{w}$
\State else
\State $\Delta_j\stackrel{\$}{\leftarrow}\mathbb{G}$

\EndIf


\State Append $\DOID'$ to ${\bf \mathsf{ISet}[\ell]}$ 


\State Append $\Delta$ into ${{ \mathsf{FSet}}[r_{{\DOID}_i}]}$
\EndFor

\State $\mathbf{W}[\mathsf{w}] \leftarrow\left(\mathsf{S T}_{c}, c\right)$
\EndFor

\end{algorithmic}


\underline{H(k,st)}
 \begin{algorithmic}[1]
 \State v$\gets$H(k,st)
 \If{v$=\perp$}
 \State v$\stackrel{\$}{\leftarrow}\{0,1\}^{\lambda}$
 \color{blue}
 \If{$\exists\mathsf{w},c$ s.t. st=$\mathsf{ST}_c\in \mathbf{W}[\mathsf{w}]$}
 \State bad$\gets$true; v$\gets\ell_{ij}$
 \EndIf
 \color{black}
 \State H(k,st)$\gets$v
 \EndIf
 \State Return v
 \end{algorithmic}


\end{algorithm}


\textbf{Game} G$_3\hspace{4mm}$In G$_3$, in the Update phase, instead of calling H to generate the $\ell$, we
pick random strings. Then, during the Search protocol, the random oracle H is programmed so that
H(K1, $\mathsf{tk}^{(\mathsf{ST_{c}}\text{ mod }p)}$) = $\ell$.
Algorithm \ref{alg: game} and \ref{alg: game G3-continue} formally describes G$_3$, and also introduces an intermediate game in blue color. In the pseudo-code, we explicit the calls to the random oracle H, and keep track of the transcripts via the table H.

The point of $\hat{\text{G}_3}$ is to ensure consistency of H's transcript: in $\hat{\text{G}_3}$, H is never programmed to two different values for the same input by Search' line 8. Instead of immediately generating the $\ell$ derived from the $c$-th $S T$ for keyword $w$ from H, $\hat{\text{G}_3}$ randomly either chooses them if $\left(\mathsf{S T_{c}}\right)$ does not already appear in H's transcript, or, if this is already the case, sets $\ell$ to the already chosen value $\mathrm{H}\left[K1,\mathsf{g^{ST'_{c}\cdot tag_w}}\right]$. Then, $\hat{\text{G}_3}$ programs the random oracle when needed by the Search protocol (line 8) or by an adversary's query (line 5 of H), so that it's outputs are consistent with the chosen values of the $\ell$'s.

By using query-info and getting the information for IDs that are going to be deleted with their keywords that will be searched before deletion (getting the information of Delw in advance), the entries are generated honestly as they are going to be revealed later, and for the not-deleted, not-searched entries, the entries look independent random (line 24). If the adversary is able to distinguish these two games, we can use it to distinguish problem D-ACE.
We can use Algorithm \ref{alg: assumption2} to simulate all the entries to the adversary.
The $a_i,b_j,c_j$ in D-ACE correspond to $\mathsf{tag_w,ST'_c,tag_\DOID}$ in the G$_3$, respectively. 
\begin{center}
    $\mathbb{P}\left[G_2=1\right]-\mathbb{P}\left[\hat{\text{G}_3}=1\right] \leq \operatorname{Adv}_{B_4}^{\mathrm{D-ACE}}(\lambda)$
\end{center}



\begin{algorithm}[ht!]
\caption{\textbf{Game G$_3$, \textcolor{blue}{$\hat{\text{G}_3}$}}-continue}
        \label{alg: game G3-continue}

\underline{Update-del (all entries for a particular $\DOID_i$)}
 \begin{algorithmic}[1]

\State Use $\mathsf{tag}_{{\DOID}_i}$, $r_{{\DOID}_i}$ 
\For {all elements $\Delta_i$ in $\mathsf{FSet[r_{{\DOID}_i}]}$ in order}
\State Compute $\ell \gets H(k1,{\Delta_i}^{{\mathsf{tag}_{{\DOID}_i}}})$
\State Remove corresponding entry from $\mathsf{ISet}[\ell]$ and $\mathsf{\ell}$
\EndFor
\State Remove entries of $\mathsf{FSet[r_{{\DOID}_i}}]$ and $\mathsf{r_{{\DOID}_i}}$
\end{algorithmic}

\underline{Search}
\begin{algorithmic}[1]

\State Use $\mathsf{tag}_{\mathsf{w}}$, $\mathsf{tk}\gets \mathsf{g}^{\mathsf{tag}_{\mathsf{w}}}$ and gets $\left(\mathsf{S T}_{c}, c\right) \leftarrow \mathbf{W}[\mathsf{w}]$
\State $\mathsf{RSet}$ $\gets \{\}$

\If {$\left(\mathsf{S T}_{c}, c\right)=\perp$}
\State  return $\emptyset$
\EndIf
\State Send $\left(\mathsf{tk}, \mathsf{S T}_{c}, c\right)$ to the server.\newline
Server:
\For {$i=c$ to 1}
\State $ \ell \leftarrow H(k1,\mathsf{tk}^{({\mathsf{S T}_{c}}\text{ mod }p)})$
\State $ \DOID' \leftarrow \mathsf{ISet}\left[\ell\right]$
\State $\mathsf{RSet}$ $\gets \mathsf{RSet}\cup \DOID'$
\State $\mathsf{S T}_{i-1} \leftarrow \pi_{\mathrm{PK}}\left(\mathsf{S T}_{i}\right)$
\EndFor

\State \Return{$\mathsf{RSet}$}
\end{algorithmic}

\end{algorithm}



To bound the distinguishing advantage between $\hat{\text{G}_3}$ and ${\text{G}_3}$, we can see that, if bad is set to true, we can break the one-wayness of the trapdoor permuattion (TDP). More formally, we can construct a reduction $B_3$ from a distinguisher A inserting N keyword/document pairs in the database (refer to \cite{bost2016ovarphiovarsigma} for more information).

\begin{center}
    $\mathbb{P}\left[\hat{\text{G}_3}=1\right]-\mathbb{P}\left[G_3=1\right] \leq N \cdot \operatorname{Adv} ^{\mathrm{OW}}_{\pi, B_3}(\lambda)$
\end{center}

Therefore,

\begin{center}
    $\mathbb{P}\left[G_2=1\right]-\mathbb{P}\left[G_3=1\right] \leq N \cdot \operatorname{Adv} ^{\mathrm{OW}}_{\pi, B_3}(\lambda)+\operatorname{Adv}_{B_4}^{\mathrm{D-ACE}}(\lambda)$
\end{center}

\noindent
\textbf{Game} G$_4\hspace{4mm}$In Search, G$_4$ generates the search token from $\mathsf{ST_0}$ by iterating $\Pi$ instead of using an already computed and stored token and if an entry is accessed for the first time, the game randomly picks it in $\mathcal{M}$. This happens for all $\mathsf{ST}$s except the ones that have been used for the tags related to the query-info IDs.
\begin{center}
    $\mathbb{P}\left[G_3=1\right]-\mathbb{P}\left[G_4=1\right] =0$
\end{center}
\textbf{The simulator}$\hspace{4mm}$
The simulator is described in Algorithm \ref{alg: Simulator}. Instead of the keyword $w$, Simulator uses the counter $w$ = min sp(w) uniquely mapped from $w$ using the leakage function.
\begin{center}
    $\mathbb{P}\left[G_4=1\right]-\mathbb{P}\left[\operatorname{Ideal}_{\mathcal{A}, \mathcal{S}}^{\Sigma}(\lambda)=1\right] =0$
\end{center}

\remove{

\begin{algorithm}[hbt!]
\caption{\textbf{Game G$_4$}}
        \label{alg: game G4}
\underline{$\operatorname{Setup}\left(\lambda\right)$}
 \begin{algorithmic}[1]
\State Initialise empty maps $\mathsf{FSet,ISet,W}$

\State Select ($\mathsf{SK, PK}$) for $\pi $ 
using security parameter $\lambda$, and $\mathbb{G}$ a group of prime order $p$ and generator $\mathsf{g}$. 

\State Send $\mathsf{FSet,ISet}$ as $\mathsf{EGDB1,EGDB2}$ to the server.

\end{algorithmic}


\underline{Update-add$\left(add,(i d1,id2,\ldots),\text{ query-info}\right)$}

 \begin{algorithmic}[1]

\State Extract the timestamp of adding the $id$s from $\operatorname{AddHist}(\text{set of }id)$ and choose $u \gets \operatorname{AddHist}(\text{set of }id)$

\For{i=1 to $\mathsf{N}_\DOID$}
\State Randomly pick index $r_{{\DOID}_i}\stackrel{\$}{\leftarrow} \{0,1\}^\lambda$ and $\mathsf{tag}_{{\DOID}_{i}}\stackrel{\$}{\leftarrow} \{0,1\}^\lambda$

\For{j=1 to $\mathsf{NW}_{{\DOID}_i}$}
\State $\ell_{ij}\stackrel{\$}{\leftarrow} \{0,1\}^\lambda$
\If{$\DOID_i$ is in query-info to be deleted and $\mathsf{w}$ related to the $\Delta_j$ $\in S_i \text{ with }\{t_\text{Srch}<t_\text{Del}\}$ }
\State $\mathsf{tag_w}\stackrel{\$}{\leftarrow} \{0,1\}^\lambda$ and keep it for this $\mathsf{w}$
\State $\left(\mathsf{S T}_{c}, c\right) \leftarrow \mathbf{W}[\mathsf{w}]$
\If {$\left(\mathsf{S T}_{c}, c\right)=\perp$}
\State $\quad \mathsf{S T_{0}} \stackrel{\$}{\leftarrow} \mathcal{M}, c \leftarrow0$

\EndIf
\State $c\gets c+1$
\State $ \mathsf{S T}_{c} \leftarrow \pi_{\mathsf{S K}}^{-1}\left(\mathsf{S T}_{c-1}\right)$; $ \mathsf{S T'}_{c} \leftarrow \left(\mathsf{S T}_{c}\text{ mod }p\right)$
\State $\Delta_j \gets {\mathsf{g}^{{({\mathsf{S T'}_{c}}\cdot {\mathsf{tag}_{\mathsf{w}}}})/\mathsf{tag_{\DOID_i}}}}$// meaning: $\Delta_j \gets {{(\text{generated token})}^{{1}/\mathsf{tag_{\DOID_i}}}}$
\State program H s.t. H(k1,$\mathsf{g}^{{\mathsf{S T'}_{c}}\cdot {\mathsf{tag}_{\mathsf{w}}}}$)$\gets \ell_{ij} $
\State $\mathbf{W}[\mathsf{w}]\gets\left(\mathsf{S T}_{0}\right) $
\State else

\State $\Delta_j \stackrel{\$}{\leftarrow}\mathbb{G}$
\EndIf
\State Append $\Delta_j$ to $\mathsf{FSet}[r_{{\DOID}_i}]$



\State $\DOID' \gets E(\mathsf{K_w},\{0\}^\lambda)$// ($\mathsf{K_w}$ generated randomly for each $\mathsf{w}$ and kept in a set for later use)
\State Append $\DOID'_{ij}$ to $\mathsf{ISet}[\ell_{ij}]$
\EndFor
\EndFor

\remove{
\State extract the $\DOID$ with the same $\operatorname{AddHist}(id)$ and use $\operatorname{N}(add)$ to generate the below tags to add to the database

\For {each ${\mathsf{w}}$} 
\State get $\mathsf{tag}_{\mathsf{w}}$ and $\K_{\mathsf{w}} $ from M0.//specific ${\DOID}_i$

\State $\left(\mathsf{S T}_{c}, c\right) \leftarrow \mathbf{W}[\mathsf{w}][u-1]$
\If {$\left(\mathsf{S T}_{c}, c\right)=\perp$}
\State $\quad \mathsf{S T_{0}} \stackrel{\$}{\leftarrow} \mathcal{M}, c \leftarrow-1$
\State else
\State $\quad \mathsf{S T}_{c+1} \leftarrow \pi_{\mathsf{S K}}^{-1}\left(\mathsf{S T}_{c}\right)$
\EndIf

\State Get $\mathsf{tag}_{{{\DOID}_i}}$ and ${r_{{\DOID}_i}}$ from M1

\If {there is no index $r_{\DOID}$ in $\mathsf{FSet}$ for ${\DOID}_i$}
\State Compute index $r_{{\DOID}_i}\stackrel{\$}{\leftarrow} \{0,1\}^\lambda$
and a tag $\mathsf{tag}_{{\DOID}_i} \stackrel{\$}{\leftarrow} \{0,1\}^\lambda$
\EndIf
\State Compute $\DOID' \gets {E}(\K_{\mathsf{w}},\DOID)$ 


\State $\ell \leftarrow \mathsf{g}^{{\mathsf{S T}_{c+1}}\cdot {\mathsf{tag}_{\mathsf{w}}}}$

\State Append $\DOID'$ to ${\bf \mathsf{ISet}[\ell]}$ and $c\gets c+1$\newline

\State Compute $\Delta \gets \mathsf{g}^{{\mathsf{S T}_{c+1}}\cdot {\mathsf{tag}_{\mathsf{w}}}/{\mathsf{tag}_{{\DOID}_i}}}$ 



\State Append $\Delta$ into ${{ \mathsf{FSet}}[r_{{\DOID}_i}]}$
\EndFor
\State $\mathbf{W}[\mathsf{w}][u] \leftarrow\left(\mathsf{S T}_{c+1}, c+1\right)$

\State $\mathsf{E}\DBGePh1[u] =\mathsf{FSet}, \mathsf{E}\DBGePh2[u] = \mathsf{ISet}$.
}

\end{algorithmic}




\end{algorithm}

\begin{algorithm}[hbt!]
\caption{\textbf{Game G$_4$}-continue}
        \label{alg: game G4-continue}

\underline{Update-del$\left(del,id\right)$}
 \begin{algorithmic}[1]

\State Extract the timestamps of adding/deleting the $id$ from $\operatorname{DelHist}(id)$ and choose $u \gets u^{del} \text{ in } \operatorname{DelHist}(id)$

\State Extract the random chosen $\mathsf{tag}_{{\DOID}_i}$, and use ${r_{{\DOID}_i}}$ for deleting $\DOID_i$
\For {all elements $\Delta_j$ in $\mathsf{FSet[r_{{\DOID}_i}]}$, use the extracted correlations ($\Delta_j 2\ell_{ij}$ in $\mathsf{Delindex}(id))$}

\State program H s.t. H(k1, ${\Delta_{ij}}^{\mathsf{tag}_{{\DOID}_i}}$)$\gets\ell_{ij}$


\EndFor
\State Send $\mathsf{tag}_{{\DOID}_i}$, and ${r_{{\DOID}_i}}$ as deletion tokens to server
\end{algorithmic}


\underline{$\operatorname{Search}(w)$}
\begin{algorithmic}[1]
%

\State Randomly select $\mathsf{tag_w}$ or use it if w was in the S$_i$ with $\{t_{\text{Srch}}<t_{\text{Del}}\}$ in Update
\State $\mathsf{S T_{0}} \stackrel{\$}{\leftarrow} \mathcal{M}$ for the ones not in S$_i$, and $c\gets 1$; $(\mathsf{ST}_{c},c) \leftarrow \mathbf{W}[{\mathsf{w}}]$ for $w\in S_i$

\For{all added $\DOID$s (m number of them) in $\operatorname{rp}(w)$ at time u in comparison with $\operatorname{rp}(w)$ at time u-1}

\For{i=c to c+m-1}
\State skip the skipped tokens from the leakage (indices got deleted before being searched) by computing $\mathsf{S T}_{i} \leftarrow \pi_{\mathsf{S K}}^{-1}\left(\mathsf{S T}_{i-1}\right)$
\State Compute $\mathsf{S T}_{i} \leftarrow \pi_{\mathsf{S K}}^{-1}\left(\mathsf{S T}_{i-1}\right)$ for non-deleted ones
\State $ \mathsf{S T'}_{i} \leftarrow \left(\mathsf{S T}_{i}\text{ mod }p\right)$
\State program H s.t. H(k1,$\mathsf{g}^{{\mathsf{S T'}_{i}}\cdot {\mathsf{tag}_{\mathsf{w}}}}$)$\gets \ell_{ij} $

\EndFor
\EndFor

\State $\mathsf{ST}_{m} \leftarrow \mathbf{W}[{\mathsf{w}}]$

\State Send $\left(\mathsf{g^{tag_{w}}}, \mathsf{S T}_{m}\right)$ to the server.

\end{algorithmic}

\end{algorithm}

}

%

\begin{algorithm}[hb!]
\caption{\textbf{Simulator}}
        \label{alg: Simulator}
\underline{$\operatorname{Setup}\left(\mathcal{L}^{S t p}(\lambda)\right)$}
 \begin{algorithmic}[1]
\State Initialise empty maps $\mathsf{FSet,ISet,W}$

\State Select ($\mathsf{SK, PK}$) for $\pi $ 
using security parameter $\lambda$, and $\mathbb{G}$ a group of prime order $p$ and generator $\mathsf{g}$. 

\State Send $\mathsf{FSet,ISet}$ as $\mathsf{EGDB1,EGDB2}$ to the server.

\remove{
\For {each ${\mathsf{w}}$} 
\State $\mathsf{tag}_{\mathsf{w}} \stackrel{\$}{\leftarrow} \{0,1\}^\lambda$; $\K_{\mathsf{w}} \stackrel{\$}{\leftarrow} \{0,1\}^\lambda$. keep these two in a map M0 for $\mathsf{w}$
\State Set a counter $c\gets 1$ and select $ \mathsf{S T_{0}} \stackrel{\$}{\leftarrow} \mathcal{M}$
\For {${\DOID}_i \in \DBGePh({\mathsf{w}})$}
\If {there is no index $r_{\DOID}$ in $\mathsf{FSet}$ for ${\DOID}_i$}
\State Compute index $r_{{\DOID}_i}\stackrel{\$}{\leftarrow} \{0,1\}^\lambda$
and a $\mathsf{tag}_{{\DOID}_i} \stackrel{\$}{\leftarrow} \{0,1\}^\lambda$. keep these in M1 for $\DOID_i$
\EndIf
\State Compute $\DOID' \gets {E}(\K_{\mathsf{w}},\DOID)$ 

\State $ \mathsf{S T}_{c} \leftarrow \pi_{\mathsf{S K}}^{-1}\left(\mathsf{S T}_{c-1}\right)$

\State $\ell \leftarrow \mathsf{g}^{{\mathsf{S T}_{c}}\cdot {\mathsf{tag}_{\mathsf{w}}}}$



\State Append $\DOID'$ to ${\bf \mathsf{ISet}[\ell]}$ and $c\gets c+1$\newline


\State Compute $\Delta \gets \mathsf{g}^{{\mathsf{S T}_{c}}\cdot {\mathsf{tag}_{\mathsf{w}}}/{\mathsf{tag}_{{\DOID}_i}}}$ 



\State Append $\Delta$ into ${{ \mathsf{FSet}}[r_{{\DOID}_i}]}$
\EndFor
\State $\mathbf{W}[\mathsf{w}][u] \leftarrow\left(\mathsf{S T}_{c}, c\right)$
\EndFor
\State $\mathsf{E}\DBGePh1[u] =\mathsf{FSet}, \mathsf{E}\DBGePh2[u] = \mathsf{ISet}$.
\State \Return{$\mathsf{E}\DBGePh1,\mathsf{E}\DBGePh2$} //Stored on Server, $\mathbf{W}[\mathsf{w}]$
 and {$\K=(\K_S, \K_T, \K_1, \K_2)$} //Stored on Vetter
 }
\end{algorithmic}


\underline{Update-add$\left(\mathcal{L}^{U p d t}(add,(i d1,id2,\ldots)),\text{ query-info}\right)$}

 \begin{algorithmic}[1]

\State Extract the timestamp of adding the $id$s from $\operatorname{AddHist}(\text{set of }id)$ and choose $u \gets \operatorname{AddHist}(\text{set of }id)$

\For{i=1 to $\mathsf{N}_\DOID$}
\State Randomly pick index $r_{{\DOID}_i}\stackrel{\$}{\leftarrow} \{0,1\}^\lambda$ and $\mathsf{tag}_{{\DOID}_{i}}\stackrel{\$}{\leftarrow} \{0,1\}^\lambda$
%
\For{j=1 to $\mathsf{NW}_{{\DOID}_i}$}
\State $\ell_{ij}\stackrel{\$}{\leftarrow} \{0,1\}^\lambda$
\If{$\DOID_i$ is in query-info to be deleted and $\mathsf{w}$ related to the $\Delta_j$ $\in \{\text{Srch}<\text{Del}\}_i$ }
\State $\ell_{ij}\stackrel{\$}{\leftarrow} \{0,1\}^\lambda$ and keep it for this $\mathsf{w}$
\State $\left(\mathsf{S T}_{c}, c\right) \leftarrow \mathbf{W}[\mathsf{w}]$
\If {$\left(\mathsf{S T}_{c}, c\right)=\perp$}
\State $\quad \mathsf{S T_{0}} \stackrel{\$}{\leftarrow} \mathcal{M}, c \leftarrow0$

\EndIf
\State $c\gets c+1$
\State $ \mathsf{S T}_{c} \leftarrow \pi_{\mathsf{S K}}^{-1}\left(\mathsf{S T}_{c-1}\right)$; $ \mathsf{S T'}_{c} \leftarrow \left(\mathsf{S T}_{c}\text{ mod }p\right)$
\State $\Delta_j \gets {\mathsf{g}^{{({\mathsf{S T'}_{c}}\cdot {\mathsf{tag}_{\mathsf{w}}}})/\mathsf{tag_{\DOID_i}}}}$// meaning: $\Delta_j \gets {{(\text{generated token})}^{{1}/\mathsf{tag_{\DOID_i}}}}$
\State program H s.t. H(k1,$\mathsf{g}^{{\mathsf{S T'}_{c}}\cdot {\mathsf{tag}_{\mathsf{w}}}}$)$\gets \ell_{ij} $
\State $\mathbf{W}[\mathsf{w}]\gets\left(\mathsf{S T}_{0}\right) $
\State else

\State $\Delta_j \stackrel{\$}{\leftarrow}\mathbb{G}$
\EndIf



\State Append $\Delta_j$ to $\mathsf{FSet}[r_{{\DOID}_i}]$

\State $\DOID' \gets E(\mathsf{K_w},\{0\}^\lambda)$// ($\mathsf{K_w}$ generated randomly for each $\mathsf{w}$ and kept in a set for later use)
\State Append $\DOID'_{ij}$ to $\mathsf{ISet}[\ell_{ij}]$
\EndFor
\EndFor

\remove{
\State extract the $\DOID$ with the same $\operatorname{AddHist}(id)$ and use $\operatorname{N}(add)$ to generate the below tags to add to the database

\For {each ${\mathsf{w}}$} 
\State get $\mathsf{tag}_{\mathsf{w}}$ and $\K_{\mathsf{w}} $ from M0.//specific ${\DOID}_i$

\State $\left(\mathsf{S T}_{c}, c\right) \leftarrow \mathbf{W}[\mathsf{w}][u-1]$
\If {$\left(\mathsf{S T}_{c}, c\right)=\perp$}
\State $\quad \mathsf{S T_{0}} \stackrel{\$}{\leftarrow} \mathcal{M}, c \leftarrow-1$
\State else
\State $\quad \mathsf{S T}_{c+1} \leftarrow \pi_{\mathsf{S K}}^{-1}\left(\mathsf{S T}_{c}\right)$
\EndIf

\State Get $\mathsf{tag}_{{{\DOID}_i}}$ and ${r_{{\DOID}_i}}$ from M1

\If {there is no index $r_{\DOID}$ in $\mathsf{FSet}$ for ${\DOID}_i$}
\State Compute index $r_{{\DOID}_i}\stackrel{\$}{\leftarrow} \{0,1\}^\lambda$
and a tag $\mathsf{tag}_{{\DOID}_i} \stackrel{\$}{\leftarrow} \{0,1\}^\lambda$
\EndIf
\State Compute $\DOID' \gets {E}(\K_{\mathsf{w}},\DOID)$ 


\State $\ell \leftarrow \mathsf{g}^{{\mathsf{S T}_{c+1}}\cdot {\mathsf{tag}_{\mathsf{w}}}}$ 



\State Append $\DOID'$ to ${\bf \mathsf{ISet}[\ell]}$ and $c\gets c+1$\newline


\State Compute $\Delta \gets \mathsf{g}^{{\mathsf{S T}_{c+1}}\cdot {\mathsf{tag}_{\mathsf{w}}}/{\mathsf{tag}_{{\DOID}_i}}}$ 



\State Append $\Delta$ into ${{ \mathsf{FSet}}[r_{{\DOID}_i}]}$
\EndFor 
\State $\mathbf{W}[\mathsf{w}][u] \leftarrow\left(\mathsf{S T}_{c+1}, c+1\right)$

\State $\mathsf{E}\DBGePh1[u] =\mathsf{FSet}, \mathsf{E}\DBGePh2[u] = \mathsf{ISet}$.
}

\end{algorithmic}




\end{algorithm}

\textbf{Finally, we can conclude:}

\begin{center}
    $\mathbb{P}[\operatorname{Real}_{\mathcal{A}, \mathcal{S}}^{\Sigma}(\lambda)=1]-[\mathbb{P}\left[\operatorname{Ideal}_{\mathcal{A}, \mathcal{S}}^{\Sigma}(\lambda)=1\right] \leq $\\
    $\operatorname{Adv}_{F, B_1}^{\mathrm{prf}}(\lambda)+\operatorname{Adv}_{SE, B_2}^{\mathrm{IND-CPA}}(\lambda)+N \cdot \operatorname{Adv} ^{\mathrm{OW}}_{\pi, B_3}(\lambda)+\operatorname{Adv}_{B_4}^{\mathrm{D-ACE}}(\lambda)$
\end{center}


\begin{algorithm}[hb!]
\caption{\textbf{Simulator}-continue}
        \label{alg: Simulator-continue}

\underline{Update-del$\left(\mathcal{L}^{U p d t}(del,i d)\right)$}
 \begin{algorithmic}[1]

\State Extract the timestamps of adding/deleting the $id$ from $\operatorname{DelHist}(id)$ and choose $u \gets u^{del} \text{ in } \operatorname{DelHist}(id)$

\State Extract the random chosen $\mathsf{tag}_{{\DOID}_i}$, and use ${r_{{\DOID}_i}}$ for deleting $\DOID_i$
\For {all elements $\Delta_j$ in $\mathsf{FSet[r_{{\DOID}_i}]}$, use the extracted correlations ($\Delta_j 2\ell_{ij}$ in $\mathsf{Delindex}(id))$}






\State program H s.t. H(k1, ${\Delta_{ij}}^{\mathsf{tag}_{{\DOID}_i}}$)$\gets\ell_{ij}$


\EndFor
\State Send $\mathsf{tag}_{{\DOID}_i}$, and ${r_{{\DOID}_i}}$ as deletion tokens to server
\end{algorithmic}


\underline{$\operatorname{Search}(\mathcal{L}^{\operatorname{Srch}}(w))$}
\begin{algorithmic}[1]

\State $\bar{\mathsf{w}}\gets \text{min(sp($\mathsf{w}$))}$
\State Randomly select $\mathsf{tag_w}$ or use it if w was in the S$_i$ with $\{t_{\text{Srch}}<t_{\text{Del}}\}$ in Update
\State $\mathsf{S T_{0}} \stackrel{\$}{\leftarrow} \mathcal{M}$ for the ones not in S$_i$, and $c\gets 1$; $(\mathsf{ST}_{c},c) \leftarrow \mathbf{W}[{\mathsf{w}}]$ for $w\in S_i$

\For{all added $\DOID$s (m number of them) in $\operatorname{rp}(w)$ at time u in comparison with $\operatorname{rp}(w)$ at time u-1}

\For{i=c to c+m-1}
\State skip the skipped tokens from the leakage (indices got deleted before being searched) by computing $\mathsf{S T}_{i} \leftarrow \pi_{\mathsf{S K}}^{-1}\left(\mathsf{S T}_{i-1}\right)$
\State Compute $\mathsf{S T}_{i} \leftarrow \pi_{\mathsf{S K}}^{-1}\left(\mathsf{S T}_{i-1}\right)$ for non-deleted ones
\State $ \mathsf{S T'}_{i} \leftarrow \left(\mathsf{S T}_{i}\text{ mod }p\right)$
\State program H s.t. H(k1,$\mathsf{g}^{{\mathsf{S T'}_{i}}\cdot {\mathsf{tag}_{\mathsf{w}}}}$)$\gets \ell_{ij} $//use TimeDB[w] to extract set of $\ell$s

\EndFor
\EndFor

\State $\mathsf{ST}_{m} \leftarrow \mathbf{W}[\bar{\mathsf{w}}]$



\State Send $\left(\mathsf{g^{tag_{w}}}, \mathsf{S T}_{m}\right)$ to the server.



\end{algorithmic}

\end{algorithm}

\end{proof}

\remove{

\subsection{Part One}

\subsection{Part Two}

\section{Online Resources}

}


\end{document}